\documentclass[
10pt,aps,prx,twocolumn,superscriptaddress,nofootinbib,
amsmath,amssymb,floatfix]{revtex4-2}
\usepackage{xcolor}
\usepackage{graphicx}
\usepackage{mathtools}
\usepackage{amsthm}
\usepackage[hypertexnames=false]{hyperref}
\usepackage[capitalize, nameinlink]{cleveref}
\usepackage{enumitem}
\usepackage{microtype}
\usepackage[breakable,skins]{tcolorbox}
\usepackage{booktabs}
\usepackage{threeparttable}

\allowdisplaybreaks

\definecolor{red4}{Hsb}{0,0.79,0.72}
\definecolor{blue4}{Hsb}{240,0.65,0.85}
\definecolor{purple4}{Hsb}{330,0.79,0.70}
\hypersetup{colorlinks,
citecolor=blue4,
linkcolor=red4,
urlcolor=purple4}

\DeclarePairedDelimiter{\ket}{\lvert}{\rangle}
\DeclarePairedDelimiter{\bra}{\langle}{\rvert}
\DeclarePairedDelimiter{\norm}{\lVert}{\rVert}

\DeclareMathOperator{\diag}{diag}

\newcommand{\braket}[2]{\langle{#1}|{#2}\rangle}
\newcommand{\ketbra}[2]{\ensuremath{\ket{#1} \bra{#2}}}

\newtheorem{lemma}{Lemma}
\newtheorem{theorem}{Theorem}
\newtheorem{corollary}{Corollary}
\newtheorem{result}{Result}
\newtheorem{definition}{Definition}
\newtheorem{problem}{Problem}

\makeatletter

\newcommand{\apptocfile}{atoc}
\let\apptoc@orig@appendix\appendix
\renewcommand{\appendix}{%
  \apptoc@orig@appendix
  \let\apptoc@orig@addtocontents\addtocontents
  \long\def\addtocontents##1##2{%
    \def\apptoc@ext{##1}%
    \def\apptoc@toc{toc}%
    \ifx\apptoc@ext\apptoc@toc
      \apptoc@orig@addtocontents{\apptocfile}{##2}%
    \else
      \apptoc@orig@addtocontents{##1}{##2}%
    \fi
  }%
}
\newcommand{\appendixtableofcontents}{%
  \begingroup
    \setcounter{tocdepth}{3}%
    \phantomsection
    \let\addcontentsline\@gobblethree
    \section*{Contents}%
    \pdfbookmark[1]{Appendices}{apxcontents}%
    \@starttoc{\apptocfile}%
  \endgroup
}

\renewcommand\paragraph{\@startsection{paragraph}{4}{\parindent}
  {1ex \@plus1ex \@minus.2ex}
  {-1em}
  {\normalfont\normalsize\bfseries}}
  \makeatletter
\renewcommand\subparagraph{\@startsection{subparagraph}{5}{\parindent}
  {1ex \@plus1ex \@minus.2ex}
  {-1em}
  {\normalfont\normalsize\itshape}}

\makeatother

\crefname{section}{Sec.}{Secs.}

\begin{document}
\title{Quantum enhanced rare event discovery and sampling}

\author{Naixu Guo}
\email{naixug@u.nus.edu}
\affiliation{Centre for Quantum Technologies, National University of Singapore, Singapore 117543, Singapore}

\author{Po-Wei Huang}
\affiliation{Mathematical Institute, University of Oxford, Oxford OX2 6GG, United Kingdom}

\author{Qisheng Wang}
\affiliation{School of Computer Science, Shanghai Jiao Tong University, Shanghai 200240, China}
\affiliation{School of Informatics, University of Edinburgh, Edinburgh EH8 9AB, United Kingdom}

\author{Jayne Thompson}
\affiliation{Centre for Quantum Technologies, National University of Singapore, Singapore 117543, Singapore}
\affiliation{College of Computing and Data Science, Nanyang Technological University, Singapore 639798, Singapore}

\author{Patrick Rebentrost}
\affiliation{Centre for Quantum Technologies, National University of Singapore, Singapore 117543, Singapore}
\affiliation{School of Computing, National University of Singapore, Singapore 117417, Singapore}

\author{Mile Gu}
\email{ceptryn@gmail.com}
\affiliation{Centre for Quantum Technologies, National University of Singapore, Singapore 117543, Singapore}
\affiliation{Nanyang Quantum Hub, School of Physical and Mathematical Sciences, Nanyang Technological University, Singapore 637371, Singapore}

\author{Chengran Yang}
\email{yangchengran92@gmail.com}
\affiliation{Nanyang Quantum Hub, School of Physical and Mathematical Sciences, Nanyang Technological University, Singapore 637371, Singapore}

\date{\today}

\begin{abstract}  
Financial crashes, cascading failures in infrastructure, and critical errors in AI systems are frequently triggered by events that occur with extremely small probability. Efficiently discovering and sampling events with probability below a threshold is therefore of critical interest. Yet this task is highly non-trivial using existing classical or quantum methods. Being rare, such events require an immense sampling overhead to collect sufficient data samples. Moreover, because the rare events are not known in advance, they cannot be flagged for amplification using standard techniques. Here, we introduce a quantum algorithm for rare-event discovery and sampling without first learning which events are rare. The algorithm achieves the optimal quantum scaling with the rarity threshold.
We further demonstrate that this can achieve a quadratic speedup for heavy-tailed systems whose tail has nonvanishing total mass, and translates into a robust polynomial speedup for stationary stochastic processes, with the exponent determined by its entropy-rate structure.
\end{abstract}

\maketitle

\section{Introduction}
From edge-case scenarios that compromise AI safety~\cite{hendrycks2018baselinedetectingmisclassifiedoutofdistribution,amodei2016concreteproblemsaisafety} to black-swan and seismic events~\cite{utsu1999earthquake, shyalika2023comprehensive}, rare events can have a disproportionate impact. The discovery of such events — together with the simulation of their potential consequences — is therefore critical for accurate risk assessment and for devising contingencies to avert catastrophe~\cite{Embrechts2013Modelling,Sornette2012Dragonkings}. Consequently, there is strong interest in efficient methods to sample only from the rare events of a distribution, namely those whose probabilities fall below some target threshold $\Delta$.

Yet the scarcity of such events renders them statistically difficult to model. Critically, we do not \emph{a priori} possess a list of what outcomes are rare, nor knowledge of what effects they may have. This limits the applicability of existing quantum rare-event techniques, such as quantum Monte Carlo and amplitude amplification, which assume that the target events can be efficiently identified and flagged. Classical methods face similar challenges: importance sampling approaches~\cite{Tokdar2010Importancea,Balesdent2013Krigingbased} rely on carefully designed biasing distributions whose effectiveness depends strongly on prior knowledge of the rare events themselves. Consequently, it is not immediately clear whether methods more efficient than classical Monte Carlo sampling exist. The standard approach is therefore first to estimate the probabilities of outcomes to determine which events are rare, and only then to sample the process and post-select those events.

Here, we introduce a quantum algorithm for sampling rare events without first knowing which events are rare. Given a probability distribution $P$ through a quantum state preparation unitary $U_P$, our algorithm produces samples from the distribution restricted to events whose probabilities lie below the threshold $\Delta$. Its query complexity scales as 
$\mathcal{O}(1/\sqrt{p_{\mathrm{rare}}\Delta})$, where $p_{\mathrm{rare}}$ is the probability that an ordinary sample from $P$ is rare. 
We prove that the algorithm achieves optimal scaling with respect to $\Delta$, offering a quadratic improvement over the optimal classical $1/\Delta$ dependence.
We further demonstrate speedups in heavy-tailed distributions and stochastic processes, where many individually unlikely events can together form a substantial tail.
Beyond sampling, our algorithm can also synthesize a coherent quantum representation of the rare-event distribution that presents a crucial quantum resource for further downstream quantum advantage

\begin{figure*}
    \includegraphics[width=0.7\linewidth]{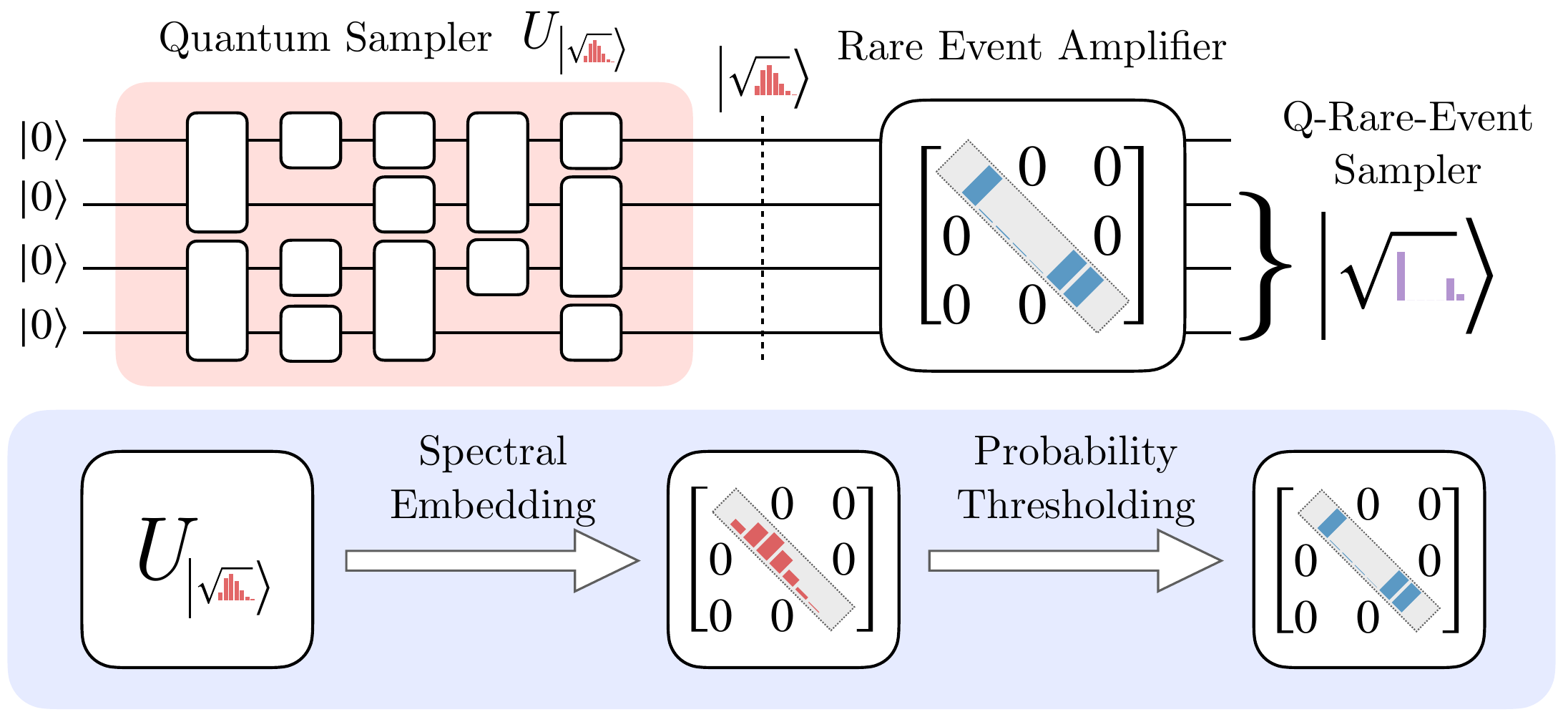}
    \caption{\textbf{Overview of quantum rare event sampling.} \textbf{Top:} Starting from the initial sample state $|P\rangle$, the rare-event amplifier filters basis states with $P(x_i) \le \Delta$ and amplifies the surviving subspace to prepare the target state $|P_{\mathcal R}\rangle$. \textbf{Bottom:} The amplifier is implemented by first converting amplitudes into a spectral representation (spectral embedding), and then applying an approximate threshold at $\sqrt{\Delta}$. This coherent pipeline enables sampling from the rare-event distribution without requiring prior identification of what events are rare.
    }
    \label{fig:alg.explanation}
\end{figure*}

\section{Results}
\paragraph*{Rare event sampling.}
Consider a probability distribution $P$ over $N$ events accessible via sampling. For a given threshold $\Delta>0$, we define the unknown but non-empty set of rare events as $\mathcal{R}=\{x_i \mid 0 < P(x_i) \leq \Delta\}$, which has a total probability mass of $p_{\mathrm{rare}}=\sum_{x_i\in \mathcal{R}}P(x_i)$.
Our objective is to construct a procedure capable of generating samples from the conditional distribution
\begin{equation}
    P_{\mathcal{R}}(x_i) := \begin{cases} \frac{P(x_i)}{p_{\text{rare}}} & \text{if } x_i \in \mathcal{R} \\ 0 & \text{if } x_i \notin \mathcal{R} \end{cases}.
\end{equation}

A sharp cutoff at $\Delta$ is unattainable with finite samples or bounded circuit depth.
Because near-threshold events become indistinguishable from the true rare set, we relax the boundary by defining an ambiguous region of near-threshold events that may be mistakenly included: $\mathcal{S}_{\Delta} = \{x_i: \Delta < P(x_i) \leq (1+\alpha) \Delta\}$ for a small constant $\alpha > 0$ (e.g., $\alpha=0.001$). This relaxation introduces an approximation error, measured by the total variation distance $D_{\mathrm{tv}}(P_\mathcal{R}, P'_\mathcal{R})=\frac{1}{2}\sum_{x_i}|P_\mathcal{R}(x_i)-P'_\mathcal{R}(x_i)|$ between the ideal and achieved distributions. This error is controlled by the relative weight $\zeta = p_{\Delta}/p_{\mathrm{rare}}$, where $p_{\Delta} = \sum_{x_i \in \mathcal{S}_{\Delta}} P(x_i)$. Our question is: how many samples from $P$ (or applications of $U_P$) are required to construct a sampler $\mathcal{O}(\zeta+\epsilon)$-close to $P_{\mathcal{R}}$ for a tunable precision $\epsilon$?

To achieve this task quantum mechanically, we assume access to a quantum sampler (unitary) $U_P$ that prepares the quantum sample state
\begin{align}
    \ket{P} = U_P \ket{0}=\sum_{i=1}^N \sqrt{P(x_i)}\ket{x_i},
\end{align}
where $\ket{x_i}$ denotes the computational basis state corresponding to the binary representation of event $x_i$, analogous to other quantum sampling algorithms \cite{aharonov2003adiabatic,gilyen2020distributional,zoufal2019quantum, temme2025quantizedmarkov}.
Our objective then translates to constructing a quantum circuit $U_{P_{\mathcal{R}}}$ that prepares the quantum rare event state
\begin{align}
    \ket{P_\mathcal{R}}=U_{P_{\mathcal{R}}} \ket{0}=\frac{1}{\sqrt{p_{\mathrm{rare}}}}\sum_{x_i\in \mathcal{R}}\sqrt{P(x_i)}\ket{x_i}.
\end{align}
Measuring the state $\ket{P}$ or $\ket{P_{\mathcal{R}}}$ in the computational basis yields samples distributed according to $P$ and $P_{\mathcal{R}}$ respectively, matching the classical definitions.
This correspondence allows us to compare classical and quantum algorithms based on the number of samples from $P$ or applications of $U_P$ required to construct the rare event sampler.
Beyond sampling, the quantum circuit $U_{P_{\mathcal{R}}}$ can efficiently prepare many copies of $\ket{P_{\mathcal{R}}}$ for use in downstream quantum algorithms where such superpositions on critical.

Our goal is to develop both classical and quantum algorithms for this task, characterize the optimal sample complexity, and identify the potential for quantum advantage.
Further details are provided in \cref{sec.algorithm}, and we summarize the problem setting as \cref{fig.problem.def}.

\begin{figure}
    \vspace*{0.5em}
    \begin{tcolorbox}[enhanced]
        \centering
        \textbf{Problem: Rare Event Sampling}\\[0.5em]
        \begin{tabular}{@{} l @{\hspace{1.0em}} p{0.62\linewidth} @{}}
            \textbf{Input:} 
            & Sampling access to $P$ classically, or quantum access through $U_P$. \\[0.3em]
            
            \textbf{Rare set:} 
            & $\mathcal{R}=\{x:0<P(x)\leq \Delta\}$. \\[0.3em]
            
            \textbf{Goal:} 
            & Construct a sampler for $P_{\mathcal{R}}$ with total variation error $\mathcal{O}(\zeta+\epsilon)$. \\[0.3em]
            
            \textbf{Cost:} 
            & Number of samples from $P$, or calls to $U_P$ and $U_P^\dagger$.
        \end{tabular}
    \end{tcolorbox}
    \caption{\textbf{Problem setting of rare event sampling.}
    The task is to sample from the conditional distribution over events whose probabilities are at most $\Delta$. The error term $\zeta$ captures the unavoidable ambiguity from near-threshold events.}
    \label{fig.problem.def}
\end{figure}

\paragraph*{Classical baseline.}
We first establish what can be achieved classically. If the set of rare events $\mathcal{R}$ were known a priori, one could directly apply rejection sampling: generate a sample $x_i$ from $P$ and accept it if and only if $x_i \in \mathcal{R}$.
Each accepted sample would then require $\mathcal{O}(1/p_{\mathrm{rare}})$ samples from $P$ on average.

However, the rare set $\mathcal{R}$ is not known a priori. To decide whether to accept a sample $x_i$, we must determine if $P(x_i) \leq \Delta$, which requires estimating outcome probabilities to precision $\mathcal{O}(\alpha\Delta)\subseteq\mathcal{O}(\Delta)$.
By adapting distribution-learning and estimation techniques from previous works \cite{Waggoner15, Apeldoorn2021multidimensional}, we can construct a list of candidate rare events and use it for rejection sampling. Combining these steps yields the following result (see \cref{sec.algorithm} for the detailed proof and construction).

\begin{result}[Classical rare event sampling]\label{main.thm.classical.alg} Given sampling access to a probability distribution $P$ over $N$ events, for any $\epsilon>0$,
    \begin{equation}
        \mathcal{O}\left(\min\left\{\frac{1}{\Delta}\log\left(\frac{N}{\epsilon}\right),\frac{1}{\Delta^2}\log\left(\frac{1}{\epsilon}\right)\right\}+\frac{1}{p_{\mathrm{rare}}}\right)
    \end{equation}
    samples from $P$ suffice to construct a sampler for a distribution $\mathcal{O}(\zeta + \epsilon)$-close to $P_{\mathcal{R}}$ in total variation distance, where $\zeta = p_{\Delta}/p_{\mathrm{rare}}$.
\end{result}

Here, the first term arises from the discovery of this list of rare events, while the second reflects the standard overhead of rejection sampling. As we will see, in many contexts, such as rare-event discovery in heavy-tailed distributions and stochastic processes, the former term dominates. This is because there can exist exponentially many rare events, ensuring that $p_{\text{rare}} \gg \Delta$. In such situations, our sampling cost minimally scales as $\Omega(1/\Delta)$. Our next result shows that there is way around this using any classical technique:

\begin{result}[Classical lower bound]
    For $\Delta>0$, any classical algorithm requires $\Omega(1/\Delta)$ samples from $P$ to construct a sampler for $P_{\mathcal{R}}$.
\end{result}

The above lower bound is robust, and continues to hold when we construct a sampler for a distribution that is close to $P_{\mathcal{R}}$ . The proof proceeds by reduction from the distribution distinguishing problem \cite{BY02}. Details can be found in \cref{sec.algorithm}. 

\paragraph*{Quantum algorithm.}
We now turn to the quantum setting.
A naive quantum approach would translate the classical strategy directly: use standard amplitude estimation
\cite{brassard2002quantum} to determine $P(x_i)$ and verify whether each outcome satisfies the rare-event criterion. This incurs a linear overhead in the sample space size $N$, negating the potential for a broad quantum advantage.
Even multidimensional quantum amplitude estimation \cite{Apeldoorn2021multidimensional}, while more efficient at identifying multiple outcomes, still requires $\mathcal{O}(1/\Delta)$ applications of $U_P$ to estimate probabilities with precision $\mathcal{O}(\Delta)$, which we prove to be suboptimal, and provides no scaling advantage.

Our algorithm takes a different approach: it directly amplifies the rare components of the superposition state $\ket{P}$ without explicitly needing to first identify which events are rare. The conceptual strategy is illustrated in \cref{fig:alg.explanation}, and we formalize our contribution as follows (see \cref{sec.algorithm} for details).

\begin{result}[Quantum rare event sampling]
    \label{resQuantumRareSample}
    Given a quantum sampler $U_P$ for a probability distribution $P$ over $N$ events, for any $\epsilon>0$,
    \begin{equation}
        \mathcal{O}\left(\frac{1}{\sqrt{p_{\mathrm{rare}}\Delta}} \log\left(\frac{1}{\epsilon}\right)\right)
    \end{equation}
    applications of $U_P$ suffice to prepare and measure a quantum state,
    yielding an event from a distribution $\mathcal{O}(\zeta+\epsilon)$-close to $P_{\mathcal{R}}$ in total variation distance, where $\zeta = p_{\Delta}/p_{\mathrm{rare}}$.
\end{result}

Our algorithm produces a quantum rare-event state, $|P_{\mathcal{R}}\rangle$, which is a proportional superposition of all possible rare events. Measuring $|P_{\mathcal{R}}\rangle$ yields a rare-event sample. The algorithm comprises three key components (detailed in \cref{sec.algorithm}):

\begin{enumerate}
    \item \emph{Rare-event amplification.} We can prepare the target rare-event state $|P_{\mathcal{R}}\rangle$ as long as we can synthesis the unitary operator
          $U_{\Pi_\mathcal{R}}$ such that $$U_{\Pi_\mathcal{R}}\ket{0}\ket{P}=\ket{0}\Pi_\mathcal{R}\ket{P}+\sqrt{1-p_{\mathrm{rare}}}\ket{1}\ket{\perp},$$
          where
          $\Pi_\mathcal{R}=\sum_{x_i\in \mathcal{R}} |x_i\rangle \langle x_i|$ and $\ket{\perp}$ is some arbitrary state.
          Measuring the ancilla and post-selecting on outcome $0$ yields $\Pi_\mathcal{R}\ket{P}/\sqrt{p_{\mathrm{rare}}}$, which occurs with probability $p_{\mathrm{rare}}$.
          To boost this success probability, we apply amplitude amplification \cite{brassard2002quantum}, requiring $\mathcal{O}(1/\sqrt{p_{\text{rare}}})$ applications of $U_{P}$ and $U_{\Pi_\mathcal{R}}$. Thus, we reduce the problem to that of synthesizing $U_{\Pi_\mathcal{R}}$.
          
    \item \emph{Probability thresholding.}
          To implement $U_{\Pi_\mathcal{R}}$, we apply an approximate Heaviside step function to the eigenvalues of the Hermitian operator $H=\sum_{x_i}\sqrt{P(x_i)}|x_i\rangle\langle x_i|$.
          The non-smooth Heaviside function is approximated by a polynomial of degree $\mathcal{O}(1/\sqrt{\Delta} \log(1/\epsilon))$, where $\epsilon$ is the approximation error.
          This approximation has a finite transition width of $\mathcal{O}(\alpha\sqrt{\Delta})$ around the threshold $\sqrt{\Delta}$, which corresponds directly to the ambiguous set $\mathcal{S}_{\Delta}$ defined earlier.
          This step requires $\mathcal{O}(1/\sqrt{\Delta} \log(1/\epsilon))$ applications of a unitary $U_H$ that block encodes $H$. Thus the synthesis $U_H$ enables rare-event amplification.
    \item \emph{Spectral embedding.} We can synthesis $U_H$ using only access to $U_P$.
          The key observation is that the amplitudes of $\ket{P}$ equal the eigenvalues of $H$, both given by $\sqrt{P(x_i)}$. This correspondence allows us to extract the spectral information directly from the quantum sampler $U_P$.
          By utilizing nonlinear amplitude transformation techniques \cite{guo2024nonlinear,rattew2023nonlinear, guo2025quantum, rattew2025accelerating, du2025gentle}, we convert the amplitude information in $\ket{P}$ into the eigenvalue structure of $U_H$ with only $\mathcal{O}(1)$ applications of $U_P$.
\end{enumerate}

Combining these three advances yields a systematic way to prepare $\ket{P_\mathcal{R}}$ using $U_p$ that aligns with the state complexity. We complement this constructive upper bound with a matching lower bound, confirming that our algorithm achieves optimal scaling in $\Delta$. The proof is established via reduction from the state and distribution discrimination problem~\cite{bennett1997strengths, belovs:lipics.esa.2019.16}, showing the minimum resources required by any quantum algorithm for this task. Details can be found in \cref{sec.algorithm}.
Similar to the classical case, the lower bound also holds for sampling from an approximate distribution.

\begin{figure*}
    \includegraphics[width=\linewidth]{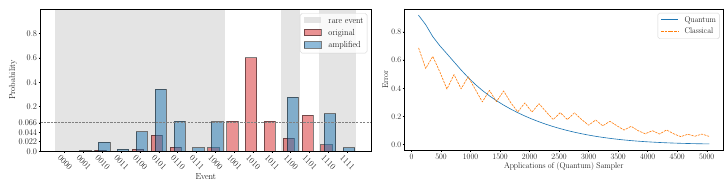}
    \caption{\textbf{Simulation results for a Dyson-Ising chain.} \textbf{Left:} Original (red) versus amplified (blue) probabilities for a sequence of length $L=4$. The dashed horizontal line represents the probability threshold $\Delta$, and the shaded regions denote the target rare events. \textbf{Right:} Total variation distance (TVD) error versus the number of applications of the quantum sampler for a sequence of length $L=8$ and Markov order $\chi=3$. The quantum algorithm converges more smoothly and rapidly than the classical baseline, which exhibits oscillatory behavior due to the intermittent sampling of rare events.}
    \label{fig:Dyson}
\end{figure*}

\begin{result}[Optimality of quantum rare event sampling]
    For $\Delta >0$, any quantum algorithm requires $\Omega(1/\sqrt{\Delta})$ applications of $U_P$ to construct a sampler for $P_{\mathcal{R}}$.
\end{result}

Importantly, the quantum algorithm produces more than a source of rare-event samples. Before the final measurement, it prepares the rare-event state $\ket{P_{\mathcal R}}$, a coherent representation of the full rare tail. This coherent form allows one to compare tails directly: for example, given rare-event states for two distributions on the same event space, a SWAP test estimates their rare-tail overlap and therefore whether the two models assign weight to similar extreme scenarios without needing to to compile a list of rare events. More broadly, rare-event state preparation upgrades rare-event sampling into a quantum primitive that can be used by downstream algorithms for analyzing distributional tails.

\paragraph*{Application to Heavy Tail Distributions.}
Comparing \cref{main.thm.classical.alg} and \cref{resQuantumRareSample}, the scaling depends on the total probability mass is carried by the rare tail as a whole. When $p_{\mathrm{rare}}=\Omega(1)$, the quantum algorithm saturates at the optimal $\mathcal{O}(1/\sqrt{\Delta})$ scaling, whereas the classical cost remains $\mathcal{O}(1/\Delta)$. The resulting quantum advantage is therefore quadratic.

Heavy-tailed distributions provide a natural setting in which rare events can be individually unlikely but collectively important. In \hyperref[sec:methods-power-law]{Methods C} and \cref{app:power_law_scaling}, we make this intuition concrete using rank-ordered power laws, where $x_i$ denotes the $i$-th most likely outcome and $P(x_i)\propto i^{-\gamma}$.
The analysis gives a transition at $\gamma=1$: for broad tails with $0<\gamma\leq 1$, the rare tail can carry non-negligible total probability mass. In these regimes, the quantum algorithm reaches its ideal  scaling, providing a quadratic improvement over the classical baseline. For steeper power laws with $\gamma>1$, the rare tail loses total mass as $\Delta$ decreases, so the additional amplification cost weakens the speedup below the ideal quadratic scaling. These models capture the central feature of fat-tailed risk in systems such as financial markets~\citep{bollerslev2013jump}, cascading infrastructure failures~\citep{dobson2007complex}, and seismic event catalogs~\citep{utsu1999earthquake}: each extreme realization may be exceptionally unlikely, yet the collection of such realizations can remain large enough to matter.

\paragraph*{Application to Stochastic Processes.} Consider a time-series governed by a stationary, ergodic stochastic processes described by a sequence of random variables $\mathfrak{X} = \ldots X_0,X_1,X_2 \ldots$. Let it have some finite entropy rate $H(\mathfrak{X})$, defined as its average uncertainty per symbol \footnote{The entropy rate is given by $H(\mathfrak{X}) = \lim_{L \to \infty} \frac{1}{L} H(X_1, X_2, \ldots, X_L)$, where $H$ is the Shannon entropy.}. The asymptotic equipartition property (AEP) tells us that almost all length-$L$ trajectories lie in a typical set, each with probability approximately $2^{-L H(\mathfrak{X})}$ as we scale $L$~\cite{shannon1948}. 

Our goal is then to sample from only rare trajectories $x_{0:L}=x_0x_1\cdots x_{L-1}$ - trajectories where \begin{equation}
P(x_{0:L}) \le \Delta = 2^{-\alpha L}
\end{equation}
for some $\alpha \geq H(\mathfrak{X})$, such that their relatively likelihood compared to typical events dies off exponentially with $L$. 

In \cref{app:ratio}, we show that this also setting in which we can determine exact - and often robust - quantum scaling advantage. Specifically, existing results in classical rare event generation tell us that such the number of trajectories such that $P(x_{0:L} = \Delta)$ generally scales as $2^{s_\Delta L}$ for some $s_\Delta > 0$~\cite{aghamohammadi2018extreme}. We can then prove that  the total probability mass of the rare events $p_{\mathrm{rare}}$, scales as $\sim \Delta^{1-\mu}$, where $\mu = s_\Delta/\alpha$. This algebraic relationship reduces the general quantum complexity of our protocol to $\tilde{\mathcal{O}}(\Delta^{(\mu -2)/2})$. 

Near full quadratic quantum advantage can be achieved when $\mu$ is near to $1$. This bound can be saturated. A particularly simple case is that of throwing $L \gg 1$ coins with some bias $p \ll 1$. In this setting, $L$ length trajectories in the typical set have probabilities scaling as $2^{-h_pL}$ where $h_p \ll 1$ is the binary entropy \footnote{That is $h_p = - p \log p - (1-p) \log (1-p)$}. We then set rare-events threshold $\Delta = 2^{-L}$. In this setting, $s_\Delta = \alpha = 1$, and thus $\mu = 1$.

\paragraph*{Numerical results.} As an illustration, we apply our algorithm to sample rare events from the thermal distribution of a Dyson--Ising spin chain with finite interaction range $\chi$. Its Hamiltonian is
\begin{equation}
    \mathcal{H}_\chi = - \sum_i \sum_{k=1}^{\chi} J_k s_i s_{i+k},
\end{equation}
where $s_i \in \{+1,-1\}$ denotes the spin at site $i$, $J_k$ is the coupling at separation $k$, and $J_k=0$ for $k>\chi$.
When considering the spins sequentially along the chain, the thermal distribution at each temperature $T$ induces a stochastic process over spin sequences. For a finite interaction range $\chi$, the conditional distribution of the next spin depends only on the previous $\chi$ spins, so the resulting process has Markov order $\chi$.

We set $\alpha=2H(\mathfrak{X})$, $\chi=3$, $J_k=2^{-k}$, $T=0.8$. We then apply our algorithm for a target $\Delta \approx 4.32 \times 10^{-3}$. In this setting, we find $\mu \approx 0.866$, signaling almost full quadratic quantum advantage. The results are shown in \cref{fig:Dyson}. 
The left panel shows that after $600$ applications of the quantum sampler, the algorithm selectively amplifies the rare events. The right panel plots the total variation distance (TVD) as a function of the application count of the quantum sampler. The TVD of the quantum algorithm decreases smoothly and more rapidly than the classical baseline. The classical curve also exhibits oscillatory behavior, as rare events are observed only intermittently. Small changes in the total sample count could distort the probability outcome if the number of rare event samples does not change. 

\section{Discussion}

We have developed a quantum algorithm for discovering and sampling rare events of a given distribution $P$, defined as those occurring with probability below a threshold $\Delta$. The algorithm prepares a quantum state superposing all such rare events weighted by their likelihood, requiring a general query complexity of $\tilde{\mathcal{O}}(1/\sqrt{p_{\mathrm{rare}}\Delta})$. When the aggregate rare-event mass $p_{\mathrm{rare}}$ is non-vanishing, this complexity saturates the optimal $\Theta(1/\sqrt{\Delta})$ scaling limit. For heavy-tailed distributions, this yields a quadratic speedup over all classical methods. Meanwhile, for stochastic processes, we can derive a general closed-form solution for $p_{\mathrm{rare}}$ and thus can analytically determine the anticipated quantum speed-up - which should be near quadratic in physically motivated settings

A key feature of our algorithm is that it produces a quantum rare-event state $|P_{\mathcal{R}}\rangle$, opening pathways beyond simply sampling rare events. Combined with the quantum SWAP test, for instance, this enables efficient measurement of whether the rare-event profiles of two distributions closely match, a quantity of significant interest in extreme value theory~\cite{blanchet2020distributionally}, financial tail risk analysis~\cite{bollerslev2013jump}, and AI safety~\cite{amodei2016concreteproblemsaisafety}.
More broadly, the ability to coherently manipulate rare-event states may serve as a primitive for quantum machine-learning methods that focus on distributional tails.

Realizing these theoretical benefits in practice requires coherent access to a quantum sampler and fault-tolerant quantum primitives. Despite these demanding hardware requirements, our algorithm interfaces naturally with existing quantum simulation subroutines. For stochastic processes in particular, recurrent quantum circuits provide systematic methods to construct the required quantum samplers~\cite{felix2018Practical,yang2018matrix,aharonov2003adiabatic, temme2025quantizedmarkov, layden2025wavefunction}, and these circuits can require drastically less memory than their classical counterparts~\cite{Yang2025Dimension}. Indeed, such memory advantages have already been demonstrated in the context of rare-event sampling for high-temperature spin chains~\cite{aghamohammadi2018extreme}. By leveraging these constructions as subroutines, our framework enables simultaneous memory and computational advantages for rare-event analysis of complex systems. Ultimately, our work demonstrates that a rigorous, optimal quantum advantage in rare-event analysis persists even when the target events cannot be identified in advance.

\section{Methods}
\label{methods}

\subsection{Classical rare event sampling}
We briefly describe the steps in the classical algorithm.
We need to estimate and construct a set of rare events $\mathcal{R}$ with the given sampling access.
There are two ways to achieve this: one is to utilize the Monte Carlo estimation with precision $\mathcal{O}(\Delta)$, and the other is to learn the distribution with $\ell_{\infty}$-distance $\mathcal{O}(\Delta)$.
The complexity of the Monte Carlo method is $\mathcal{O}(\Delta^{-1}\log N)$ with a constant success probability.
As shown in Ref.~\cite{Waggoner15}, the complexity of $\ell_{\infty}$-tomography is $\Theta(\Delta^{-2})$.
Therefore, this step requires $\mathcal{O}(\min\{\Delta^{-1}\log N, \Delta^{-2}\})$ samples.
Based on this constructed set $\mathcal{R}$, we further implement rejection sampling to achieve the sampler. For each sample $x_i$ obtained from the distribution, we check whether $x_i \in \mathcal{R}$ or not. If yes, it is our output of the rare event. Otherwise, we keep resampling. This step takes $\mathcal{O}(1/p_{\mathrm{rare}})$ samples to give a single rare event.

\subsection{Quantum rare event sampling}

Here, we provide more details about the probability thresholding and spectral embedding components of our quantum algorithm.

\subparagraph*{Probability thresholding.}
The construction requires classifying each basis state $\ket{x_i}$ as rare or non-rare based on its probability $P(x_i)$. This is equivalent to applying a filtering function:
\begin{equation*}
    \Pi_\mathcal{R} = \sum_{x_i} f\left(\sqrt{P(x_i)}\right) \ket{x_i}\bra{x_i},
\end{equation*}
where $f$ is a Heaviside step function satisfying $f(p) = 1$ if $p \leq \sqrt{\Delta}$ and $f(p) = 0$ otherwise. To implement this filtering, we observe that the amplitudes $\sqrt{P(x_i)}$ of $\ket{P}$ are precisely the eigenvalues of the Hermitian operator
\begin{equation*}
    H = \sum_{x_i} \sqrt{P(x_i)} \ket{x_i}\bra{x_i}.
\end{equation*}
The problem thus reduces to applying $f$ to the eigenvalues of $H$, a task known as eigenvalue transformation.

To implement this eigenvalue transformation on a quantum computer, we use quantum signal processing \cite{low2017optimal, gilyen2019quantum}, which applies polynomial functions to eigenvalues. We therefore approximate the Heaviside function with a polynomial. This introduces a finite transition width around the threshold $\sqrt{\Delta}$, which is a fundamental limitation since polynomials cannot implement perfectly sharp transitions. The transition width corresponds to the ambiguous set $\mathcal{S}_{\Delta}$. A polynomial of degree $d = \mathcal{O}(1/\sqrt{\Delta} \log(1/\epsilon))$ achieves approximation error $\epsilon$. Using quantum signal processing techniques, $U_{\Pi_\mathcal{R}}$ can then be constructed with $\mathcal{O}(d)$ applications of a unitary $U_H$ that embeds $H$.

\subparagraph*{Spectral embedding.}
It remains to construct $U_H$ from the quantum sampler $U_P$. Noticing that the spectral information of $H$ is already encoded in $U_P$: the first column of $U_P$ contains the amplitudes $\sqrt{P(x_i)}$, which are exactly the eigenvalues of $H$:
\begin{equation*}
    \bra{x_i}U_P\ket{0} = \sqrt{P(x_i)} = \lambda_i(H).
\end{equation*}
The challenge is to convert this amplitude information into an eigenvalue encoding coherently. Using nonlinear amplitude transformation techniques~\cite{guo2024nonlinear, rattew2023nonlinear}, which leverage quantum singular value transformation and coherent amplitude manipulation, we construct $U_H$ with only $\mathcal{O}(1)$ applications of $U_P$ and $U_P^\dagger$.

\subsection{Power-law tails}
\label{sec:methods-power-law}

We illustrate the role of the aggregate rare-event mass using rank-frequency
power-law distributions. Consider
\begin{equation*}
    P(x_k)=\frac{k^{-\gamma}}{Z_{N,\gamma}},
\end{equation*}
where $x_k$ is the $k$-th most likely event, $\gamma>0$, and $Z_{N,\gamma}=\sum_{j=1}^{N}j^{-\gamma}$. 
Since
$P(x_k)$ decreases with $k$, the rare-event set $R_\Delta=\{x_k:P(x_k)\le \Delta\}$
is a tail in rank space. Its total probability mass is
\begin{equation*}
    p_{\rm rare}
    =
    \frac{1}{Z_{N,\gamma}}
    \sum_{k=k_\Delta}^{N} k^{-\gamma},
\end{equation*}
where $k_\Delta:=\min\{k:P(x_k)\le \Delta\}$.

To expose the dependence on the support size, we set
\begin{equation*}
    N(\Delta)=\Theta(\Delta^{-q}),
\end{equation*}
and restrict to the nontrivial regime $P(x_N)\le \Delta < P(x_1)$
so that rare events exist, but not all events are rare.

As shown in \cref{app:power_law_scaling}, the rare-tail mass changes behavior at
$\gamma=1$. For broad tails with $0<\gamma<1$, the nontrivial
support-growth regime is $1\leq q< 1/(1-\gamma)$.
In this regime, $p_{\rm rare}=\Omega(1)$.
The quantum algorithm reaches
\begin{equation*}
    S_{\rm Q}(\Delta)=\mathcal{O}(\Delta^{-1/2}).
\end{equation*}
For polynomial-size support, the classical identification cost scales as
\begin{equation*}
    S_{\rm C}(\Delta)=O \left(\Delta^{-1}\log(1/\Delta)\right).
\end{equation*}
Thus broad power-law tails yield a quadratic leading-power improvement.

At the critical exponent $\gamma=1$, the same ideal quantum scaling holds
when the support grows faster than $1/\Delta$, namely when $q>1$. In that
case,
\begin{equation*}
    p_{\rm rare}\to 1-\frac{1}{q},
\end{equation*}
and again $S_{\rm Q}(\Delta)=\mathcal{O}(\Delta^{-1/2})$.
At the boundary $q=1$, the rare mass decays only logarithmically,
so the scaling is only logarithmically worse than the ideal case.

For steeper power laws with $\gamma>1$, rare events exist when
$q\ge 1/\gamma$. In this regime, $p_{\rm rare}
    =
    \Theta \left(\Delta^{(\gamma-1)/\gamma}\right),$
and hence
\begin{equation*}
    S_{\rm Q}(\Delta)
    =
    \mathcal{O} \left(\Delta^{-1+1/(2\gamma)}\right).
\end{equation*}
The algorithm still gives a polynomial improvement over the classical
baseline, but the speedup is weaker than the ideal quadratic regime because
the rare tail itself has vanishing total probability mass.

\begin{acknowledgments}
    NG, JT, PR, and MG are supported by the National Research Foundation, Singapore, through the National Quantum Office, hosted in A*STAR, under its Centre for Quantum Technologies Funding Initiative (S24Q2d0009).
    NG also acknowledges support through the Research Excellence Scholarship from SandboxAQ.
    PWH acknowledges support from the Engineering and Physical Sciences Research Council (EPSRC) Doctoral Training Partnership (EP/W524311/1), with a CASE Conversion Studentship in collaboration with Quantum Motion. PWH further acknowledges support from the Ministry of Education, Taiwan, for a Government Scholarship to Study Abroad (GSSA) and St. Catherine’s College, University of Oxford, for an Alan Tayler Scholarship. QW acknowledges the support of startup funding from Shanghai Jiao Tong University and the Engineering and Physical Sciences Research Council under grant \mbox{EP/X026167/1}. 
    MG acknowledges the support of the National Research Foundation of Singapore through the NRF Investigatorship Program (Award No. NRFNRFI09-0010). 
    CY is funded by Schmidt Sciences, LLC.
    CQT acknowledges funding from OCBC via a joint NUS-OCBC research project.
\end{acknowledgments}

\section*{Code and data availability}
The source code and data for numerics can be found at \url{https://github.com/georgepwhuang/rare_event}.

\section*{Author contributions}
NG and CY conceived the original idea. PWH conducted the numerical experiments. QW provided the proof idea for the lower bounds. 
JT and MG provided valuable feedback on the manuscript. PR provided input on the theory. All authors contributed to the theoretical analysis and the writing of the paper.

\section*{Competing interests}
All authors are inventors on a patent application related to this work.

\bibliography{main}
\clearpage

\appendix
\setcounter{theorem}{0}
\setcounter{lemma}{0}
\setcounter{corollary}{0}
\setcounter{figure}{0}
\setcounter{table}{0}
\renewcommand{\theequation}{\thesection.\arabic{equation}}
\renewcommand{\thetheorem}{S\arabic{theorem}}
\renewcommand{\theHtheorem}{S\arabic{theorem}}
\renewcommand{\thelemma}{S\arabic{lemma}}
\renewcommand{\theHlemma}{S\arabic{lemma}}
\renewcommand{\thecorollary}{S\arabic{corollary}}
\renewcommand{\theHcorollary}{S\arabic{corollary}}
\renewcommand{\thefigure}{S\arabic{figure}}
\renewcommand{\theHfigure}{S\arabic{figure}}
\renewcommand{\theproblem}{S\arabic{problem}}
\renewcommand{\theHproblem}{S\arabic{problem}}
\renewcommand{\thetable}{S\Roman{table}}
\renewcommand{\theHtable}{S\Roman{table}}
\renewcommand{\thedefinition}{S\arabic{definition}}
\renewcommand{\theHdefinition}{S\arabic{definition}}
\onecolumngrid
\begin{center}
    \noindent{\large\bfseries Appendices for ``\textsl{Quantum enhanced rare event discovery and sampling}''}
\end{center}
\appendixtableofcontents

\section{Notations}

In this appendix, we collect the notation used throughout the paper.
When we need to distinguish them, we reserve $\Omega$ for a generic finite sample space and $\mathcal{X}$ for the alphabet of a stochastic process.
A generic finite distribution is written as $P,Q:\Omega\to[0,1]$, with $\Omega=\{x_1,\dots,x_N\}$ and
\begin{equation}
    \sum_{x\in\Omega} P(x)=1, \qquad \sum_{x\in\Omega} Q(x)=1.
\end{equation}
We write $\mathbb{E}[\cdot]$ for expectation.

The total variation distance (TVD) between two probability distributions $P$ and $Q$ on the same sample space is
\begin{equation}
    D_{\mathrm{tv}}(P,Q)\coloneqq \frac12 \sum_{x\in\Omega} |P(x)-Q(x)|.
\end{equation}

For quantum states we use Dirac notation: $\ket{\psi}$ denotes a state vector, $\bra{\psi}$ its conjugate transpose, and $\langle \psi|\phi\rangle$ the inner product between $\ket{\psi}$ and $\ket{\phi}$.
A unitary $U$ satisfies $U^\dagger U = U U^\dagger = I$.
For vectors, $\|\cdot\|_2$ denotes the Euclidean norm; for matrices, $\|\cdot\|$ denotes the operator norm.

For any subset $A\subseteq \Omega$, we write
\begin{equation}
    \Pi_A \coloneqq \sum_{x\in A} \ket{x}\bra{x}
\end{equation}
for the projector onto the computational-basis subspace indexed by $A$.

Given a probability distribution $P$ over $\Omega=\{x_1,\dots,x_N\}$, a quantum sampler for $P$ is a unitary $U_P$ satisfying
\begin{equation}
    U_P\ket{0} = \sum_{i=1}^N \sqrt{P(x_i)} \ket{x_i}.
\end{equation}

We use standard asymptotic notation $\mathcal{O}(\cdot)$, $\Omega(\cdot)$, and $\Theta(\cdot)$.
We write $\widetilde{\mathcal O}(\cdot)$ to suppress polylogarithmic factors.

We denote a discrete-time stochastic process by
\begin{equation}
    \mathfrak X = \{X_t\}_{t\in\mathbb Z},
\end{equation}
where each $X_t$ takes values in a finite alphabet $\mathcal X$ of size $|\mathcal X|<\infty$.

\section{Problem definition of rare event sampling}

In this section, we formally define the rare event sampling problem.
We consider a general setting where two probability distributions, $P$ and $Q$, are defined over the same sample space.
One may understand the problem as ``Give me samples from $Q$, but only if they are unlikely under $P$."
We start by formally defining the access model.

\begin{definition}[Classical and quantum sampler]
    Let $P$ be a probability distribution over the sample space of $N$ elements $\{x_1,x_2\dots, x_N\}$, where $P(x_i)$ denotes the probability of $x_i$.
    We say we have access to a (classical) sampler for $P$ if we can sample $x_i$ according to the distribution $P$.
    We define a unitary operator $U_P$ as a quantum sampler for $P$ if
    \begin{align}
        U_P \ket{0}=\sum_{i=1}^N \sqrt{P(x_i)}\ket{x_i}.
    \end{align}
\end{definition}

Rare events are defined with respect to the distribution $P$ as elements with probabilities smaller than the threshold $\Delta>0$.
Our goal is to generate samples of these events according to the distribution $Q$.
This distinction is important for applications involving a change of measure.
For instance, if $Q$ is the uniform distribution, we sample rare events uniformly; if $Q=P$, we sample them proportional to their original probabilities.

To facilitate the lower bound proof, we distinguish between ``yes" and ``no" instances, where the ``no" instance corresponds to the case where the total probability of rare events is negligible. We introduce parameters $p_1$ and $p_2$ to separate these cases. For simplicity in later sections, we will implicitly assume $p_2 = 0$ and fix $p_1$ to a suitable value without loss of generality.

Furthermore, we define the ambiguous region as the set of events with probabilities in the interval $[\Delta, 3\Delta/2]$. We note that the upper bound $3\Delta/2$ is chosen for concreteness; this can be generalized to $(1+\alpha)\Delta$ for any constant $\alpha>0$ without affecting the validity of our results.

\begin{problem}[General rare event sampling\label{prob.rare.general}]
Assume we are given access to samplers for probability distributions $P$ and $Q$ over the same sample space of $N$ events $\{x_i\}$, respectively.
Let $\Delta>0$ and $\epsilon > 0$ be given parameters.
We define the set of rare events as $\mathcal{R}=\{x_i:P(x_i)\leq \Delta\}$ and their total probability under $Q$ as $p_{\mathrm{rare}}=\sum_{x_i\in \mathcal{R}}Q(x_i)$.
Additionally, we define the boundary region $\mathcal{S}_{\Delta}=\{x_i:\Delta< P(x_i) < 3\Delta/2\}$ and its mass $q_{\Delta}=\sum_{x_i \in \mathcal{S}_{\Delta}} Q(x_i)$.
Given $p_1,p_2\geq 0$ such that $p_1-p_2>\epsilon+q_{\Delta}$, the task is to perform the following with success probability at least $2/3$:

\begin{enumerate}
    \item If $p_{\mathrm{rare}} \geq p_1$, then construct a sampler for a distribution that is $\Theta(\epsilon+\zeta)$-close to $Q_{\mathcal{R}}$ in total variation distance, where $\zeta=q_{\Delta}/p_{\mathrm{rare}}$ and
    \begin{equation}
         Q_{\mathcal{R}}= \begin{cases}
            Q(x_{i}) /{p_{\mathrm{rare}}}, & x_{i}\in \mathcal{R},\\
            0, & x_{i}\notin \mathcal{R}.
        \end{cases}
    \end{equation} 
    \item If $p_{\mathrm{rare}} \leq p_2$, then output ``Impossible".
\end{enumerate}
\end{problem}

The rare-event sampling problem considered in the main text is recovered as a special case of \cref{prob.rare.general}.
In the generalized formulation, the distribution $P$ defines which events are rare, while the distribution $Q$ defines how these rare events should be sampled.
The main text focuses on the natural self-sampling case, where the same distribution plays both roles:
$Q=P$.
Under this specialization, the rare-event set becomes
\begin{align}\label{eq.rare.set}
    \mathcal R_\Delta
    :=
    \{x:0<P(x)\le \Delta\},
\end{align}
and its total probability mass is $p_{\mathrm{rare}}
    :=
    P(\mathcal R_\Delta)
    =
    \sum_{x\in\mathcal R_\Delta}P(x)$.
The target distribution $Q_{\mathcal R}$ in \cref{prob.rare.general} then reduces to the conditional rare-event distribution
\begin{align}
    P_{\mathcal R}(x)
    =
    \begin{cases}
        P(x)/p_{\mathrm{rare}}, & x\in\mathcal R_\Delta,\\
        0, & x\notin\mathcal R_\Delta.
    \end{cases}
\end{align}
Similarly, the ambiguous region becomes $\mathcal S_\Delta
    :=
    \{x:\Delta<P(x)\le 3\Delta/2\}$,
with mass $    p_\Delta
    :=
    P(\mathcal S_\Delta)$,
so that the unavoidable near-threshold error parameter is $\zeta
    :=p_\Delta/p_{\mathrm{rare}}.$

The parameters $p_1$ and $p_2$ are promise parameters used only in the generalized decision version of the problem.
They allow the algorithm to distinguish between the case where rare events have sufficient total mass and the case where rare-event sampling is impossible.
The main-text problem is the sampling-only, yes-instance regime, where the rare set is assumed to have nonzero mass.
Equivalently, if one wants to embed the main-text problem into \cref{prob.rare.general}, one may take $p_2=0$, $p_1=p_{\min}\leq p_{\mathrm{rare}}$, where $p_{\min}>0$ is a promised lower bound.
The additional promise condition in \cref{prob.rare.general}, $p_1-p_2>\epsilon+q_\Delta$,
is needed only if the algorithm is required to certify the ``Impossible'' case.
When the ``Impossible'' branch is omitted, as in the main text, the algorithm and its complexity are naturally stated directly in terms of the actual rare-event mass $p_{\mathrm{rare}}$.

\begin{problem}[Rare event sampling]\label{prob.rare.specific}
    Let $P$ be a probability distribution over a finite sample space
    $\Omega=\{x_1,\ldots,x_N\}$, and let $\Delta>0$ be a probability
    threshold. Define the rare-event set as \cref{eq.rare.set}, and we assume that the rare-event mass is nonzero $p_{\mathrm{rare}}>0$.
    The target rare-event distribution is the conditional distribution
    \begin{align}
        P_{\mathcal R}(x)
        :=
        \begin{cases}
            P(x)/p_{\mathrm{rare}}, & x\in\mathcal R_\Delta,\\
            0, & x\notin\mathcal R_\Delta.
        \end{cases}
    \end{align}
    Define the ambiguous region $\mathcal S_\Delta
        :=
        \{x\in\Omega:\Delta<P(x)\le (1+\alpha)\Delta\}$,
    with $\alpha>0$ a constant and total mass $p_\Delta
        :=
        P(\mathcal S_\Delta)$.
    We write $\zeta
        :=
        p_\Delta/p_{\mathrm{rare}}$
    for the relative near-threshold ambiguity.
    Given a precision parameter $\epsilon>0$, the task is to construct a
    sampler whose output distribution $\widetilde P_{\mathcal R}$ satisfies $        D_{\mathrm{tv}}
        \left(
            \widetilde P_{\mathcal R},
            P_{\mathcal R}
        \right)
        =
        \mathcal O(\epsilon+\zeta)$.
\end{problem}

\section{Preliminary to quantum linear algebra}

\label{app: amplitude block encoding}

In this section, we introduce quantum linear algebra, including block encoding, amplitude encoding, and quantum singular value transformation.

Block encoding is an encoding scheme that embeds a general matrix $A$ into a block of a unitary matrix.
Without loss of generality, we focus on the top left block of the unitary matrix $U$, that is, $(\bra{0^a}\otimes I_n) U (\ket{0^a}\otimes I_n)$.
Here, $a$ is the number of required ancillary qubits.

\begin{definition}[Block encoding~\cite{chakraborty2019power, gilyen2019quantum}\label{def.blockencoding}]
    We say a unitary $U_A$ is an $(\alpha_{\mathrm{BE}},a,\epsilon)$-encoding of matrix $A\in \mathbb{C}^{2^n\times 2^n}$ if
    \begin{align}
        \|A-\alpha_{\mathrm{BE}} (\bra{0^a}\otimes I_n)U_A(\ket{0^a}\otimes I_n)\|\leq \epsilon.
    \end{align}
\end{definition}

In our work, we focus on a specific type of block encoding, so-called the amplitude block encoding, the unitary matrix containing a diagonal matrix $A$ filled with the amplitude of a given quantum state, namely, $A = \diag(\psi_1, \dots, \psi_{N})$.
Given access to a quantum circuit $U_{\psi}$ that prepares the quantum state $\ket{\psi}=\sum_{j=1}^{N} \psi_j \ket{j}$, the amplitude block encoding can be efficiently constructed. 
Since we focus on the stochastic process, without loss of generality, we consider amplitudes as real values.

\begin{theorem}[Amplitude block encoding \cite{guo2024nonlinear, rattew2023nonlinear}\label{thm.amplitudes.encoding}]
    Given access to a $n$-qubit quantum circuit $U_{\psi}$ that prepares an $n$-qubit state $U_{\psi}:\ket{0}\rightarrow\ket{\psi}=\sum_{j=1}^{N} \psi_j \ket{j}$, where $\{\psi_j\}$ are real and $\norm{\psi}_2=1$,
    one can construct an $(1,n+2,0)$-encoding of the diagonal matrix $A=\diag(\psi_1, \dots, \psi_{N})$ with $\mathcal{O}(n)$ circuit depth and $\mathcal{O}(1)$ queries to controlled-$U_\psi$ and controlled-$U_\psi^\dagger$.
\end{theorem}

Once given a block encoding $U_A$ of a Hermitian matrix $A$, we can construct a unitary operator that is a block encoding of applying certain polynomials $g$ to the matrix $A$.

\begin{equation}
    \text{QSVT}(\Phi, U_A) = \begin{bmatrix}
        g(A)  & \cdot \\
        \cdot & \cdot
    \end{bmatrix}
\end{equation}
where $\Phi$ is a set of phase angles required to implement the polynomial $g$. In general, that is, when $A$ is not a Hermitian matrix, $g(A)$ is a matrix that is obtained by applying the polynomial $g$ to the singular values of $A$. The circuit depth scales with the degree of the polynomial. More formally, we have the following result.

\begin{lemma}[Polynomial eigenvalue transformation \cite{gilyen2019quantum}\label{theorem.qsvt}]
    Given $U$ that is an $(\alpha,a,\epsilon)$-encoding of a Hermitian matrix $A$, and a real $d$-degree polynomial $g(x)$ with $|g(x)|\leq \frac{1}{2}$ for $x\in [-1,1]$. Let $\delta>0$, one can prepare a $(1,a+n+4,4d\sqrt{\epsilon/\alpha}+\delta)$-encoding of $g(A/\alpha)$ by using $\mathcal{O}(d)$ queries to $U_A$ and $\mathcal{O}(d(a+1))$ one- and two-qubit quantum gates.
    The description of the quantum circuit can be computed classically in time $\mathcal{O}(\mathrm{poly}(d,\log(1/\delta)))$.
\end{lemma}

Therefore, we can implement many useful functions on matrices by finding their good polynomial approximations.
Here, ``good'' means that the degree of the polynomial scales logarithmically with the error in the interval $[-1,1]$.
In general, especially for functions that are not smooth in this interval, such polynomial approximations do not exist.
However, a good approximation may still be found on the subset of $[-1,1]$.
In the following, we list some known results for the sign and rectangle functions.

The rectangle function $f(x)$ is defined as
\begin{equation}
    f_t(x) \coloneqq \begin{cases}
        1, & x \in [-t,t]               \\
        0, & x \in [-1, -t)\cup (t, 1].
    \end{cases}
\end{equation}

\begin{lemma}[Polynomial approximation of the rectangle function \cite{gilyen2019quantum}\label{approx.rectangle}]
    For any error $\epsilon \in \left(0, 1/2\right)$, there exist an even $d$-degree polynomial $g_d \in \mathbb{R}[x]$ that approximates the rectangular function $f_t(x)$, such that $0\leq g_d(x) \leq 1$, $\forall x \in [-1, 1]$, and
    \begin{align}
        |f_t(x) -g_d(x)| \leq \epsilon, \quad \forall x \notin [-t-\Gamma, -t+ \Gamma]\cup [t - \Gamma, t + \Gamma].
    \end{align}
    where $\Gamma \in (0,1/2)$ denotes the size of ambiguity range, and the degree $d$ scales as  $\mathcal{O}\left(\frac{1}{\Gamma}\log\left(1/\epsilon\right)\right)$.
\end{lemma}

\section{Quantum and classical algorithm for rare event sampling\label{sec.algorithm}}

Here, we provide the classical and quantum algorithm for \cref{prob.rare.general}.
We also establish query complexity lower bounds for both cases, demonstrating that our quantum algorithm is optimal with respect to the dominant factor defining the rare event. Furthermore, we show a quadratic separation between the classical and quantum complexities.

\subsection{Classical algorithm \label{app.alg.classical}}

\begin{theorem}[Classical algorithm for general rare event sampling]\label{thm.alg.classical}
    There exists a classical algorithm that solves \cref{prob.rare.general} using
    \begin{equation}
        \mathcal{O}\left(\min\left\{ \frac{1}{\Delta}\log\left(\frac{N}{\epsilon}\right), \frac{1}{\Delta^2}\log\left(\frac{1}{\epsilon}\right) \right\} \right)
    \end{equation}
    samples from $P$ and $\mathcal{O}(\max\{1/(p_1-p_2-q_{\Delta})^2, 1/p_{\mathrm{rare}}\})$ samples from $Q$.
\end{theorem}

\begin{proof}
We first construct a set of candidate rare events, which will be used to distinguish between the two cases and subsequently for rejection sampling.
We now show that the candidate rare set can be constructed using
$O(\Delta^{-1}\log(N/\delta))$ samples from $P$. The point is that
we do not need additive $\ell_\infty$-estimation of the whole
distribution. It suffices to perform threshold classification, allowing
arbitrary behavior in the ambiguous region.

Draw $T$ independent samples from $P$, and let
\begin{align}
    \widehat P_T(x)
    =
    \frac{1}{T}
    \sum_{t=1}^T \mathbf 1\{X_t=x\}
\end{align}
be the empirical frequency of event $x$. Define the candidate set
\begin{align}
    \widehat R
    :=
    \left\{
        x:\widehat P_T(x)\le \frac{4\Delta}{3}
    \right\}.
\end{align}

First, fix $x$ with $P(x)\le \Delta$. Let
$N_x=T\widehat P_T(x)$. Then $N_x$ is binomial with mean
$TP(x)$. Since the event
$\widehat P_T(x)>4\Delta/3$ is increasing in $P(x)$, its
probability is maximized over $P(x)\le\Delta$ at $P(x)=\Delta$.
Therefore, by the multiplicative Chernoff bound,
\begin{align}
    \Pr \left[
        \widehat P_T(x)> \frac{4\Delta}{3}
    \right]
    \le
    \Pr \left[
        \mathrm{Bin}(T,\Delta)>
        \frac{4T\Delta}{3}
    \right]
    \le
    \exp \left(-\frac{T\Delta}{27}\right).
\end{align}
Thus a truly rare event is falsely excluded with probability at most
$\exp(-T\Delta/27)$.

Second, fix $x$ with $P(x)\ge 3\Delta/2$. The event
$\widehat P_T(x)\le 4\Delta/3$ is a lower-tail event, and its
probability is maximized when $P(x)=3\Delta/2$. By the Chernoff bound, we have
\begin{align}
    \Pr \left[
        \widehat P_T(x)\le \frac{4\Delta}{3}
    \right]
    \le
    \Pr \left[
        \mathrm{Bin} \left(T,\frac{3\Delta}{2}\right)
        \le \frac{4T\Delta}{3}
    \right]\leq \exp \left(-\frac{T\Delta}{108}\right).
\end{align}
Therefore, a non-rare event with $P(x)\ge 3\Delta/2$ is
falsely included with probability at most $\exp(-T\Delta/108)$.

Taking a union bound over all $N$ events, the probability that any
decisive event is misclassified is at most
\begin{align}
    N\exp \left(-\frac{T\Delta}{27}\right)
    +
    N\exp \left(-\frac{T\Delta}{108}\right)
    \le
    2N\exp \left(-\frac{T\Delta}{108}\right).
\end{align}
Thus, choosing $    T
    =
    O \left(
        1/\Delta\log(N/\delta)
    \right)$
ensures that $R_\Delta\subseteq \widehat R\subseteq R_\Delta\cup S_\Delta$
with probability at least $1-\delta$,
where $R_\Delta=\{x:P(x)\le \Delta\}$ and $S_\Delta=\{x:\Delta<P(x)\le 3\Delta/2\}$.

An alternate approach to this is to first learn the distribution directly without the two-level sampling to precision $\Delta/6$ in the $\ell_{\infty}$ distance with probability $1 - \delta$ with sample complexity $\Theta(\frac{1}{\Delta^2}\log(\frac{1}{\delta}))$ by the algorithm in Ref.~\cite{Waggoner15}, which is also able to give us the set $\widehat R$ such that $\mathcal R_\Delta\subseteq \widehat R\subseteq \mathcal R_\Delta\cup\mathcal S_\Delta$ with probability $1-\delta$ by taking
    \begin{align} \label{eq:E2}
        T=\mathcal{O}\left(\frac{1}{\Delta^2}\log\left(\frac{1}{\delta}\right)\right).
    \end{align}

Combining these two results, it suffices to take
    \begin{align}
        T = \mathcal{O}\left( \min\left\{ \frac{1}{\Delta}\log\left(\frac{N}{\delta}\right), \frac{1}{\Delta^2}\log\left(\frac{1}{\delta}\right) \right\} \right).
    \end{align}

Now we consider distinguishing between the two cases
$p_{\mathrm{rare}}\ge p_1$ and $p_{\mathrm{rare}}\le p_2$.
On the good event of the threshold-classification step, the constructed
candidate set has the form $\widehat R=\mathcal R_\Delta\cup \widetilde{\mathcal S}_\Delta$, where $\widetilde{\mathcal S}_\Delta\subseteq \mathcal S_\Delta$.

Let $\widetilde q_\Delta
    :=
    Q(\widetilde{\mathcal S}_\Delta)$, we have $0\le \widetilde q_\Delta\le q_\Delta$.
If $p_{\mathrm{rare}}\ge p_1$, then
\begin{align}
    Q(\widehat R)
    \ge
    p_{\mathrm{rare}}
    \ge
    p_1.
\end{align}
If $p_{\mathrm{rare}}\le p_2$, then
\begin{align}
    Q(\widehat R)
    \le
    p_2+\widetilde q_\Delta
    \le
    p_2+q_\Delta.
\end{align}
Thus the gap between the two cases is at least
\begin{align}
    g:=p_1-p_2-q_\Delta.
\end{align}
By the promise condition, $g>0$. Estimating $Q(\widehat R)$ to
constant accuracy relative to this gap requires
\begin{align}
    O \left(\frac{1}{g^2}\right)
    =
    O \left(\frac{1}{(p_1-p_2-q_\Delta)^2}\right)
\end{align}
samples from $Q$. Equivalently, since the promise gives
$p_1-p_2-q_\Delta>\epsilon$, one may state this step as requiring
$O(1/\epsilon^2)$ samples from $Q$.

If the first case holds, we then use rejection sampling based on
$\widehat R$. The accepted distribution is
\begin{align}
    \widetilde Q(x)
    =
    \begin{cases}
        \dfrac{Q(x)}{p_{\mathrm{rare}}+\widetilde q_\Delta},
        & x\in \mathcal R_\Delta\cup\widetilde{\mathcal S}_\Delta,\\
        0,
        & \text{otherwise}.
    \end{cases}
\end{align}
The acceptance probability is
$p_{\mathrm{rare}}+\widetilde q_\Delta$, so one accepted sample requires
$O(1/(p_{\mathrm{rare}}+\widetilde q_\Delta))\in O(1/p_{\mathrm{rare}})$
samples from $Q$.

The total variation distance between $\widetilde Q$ and the ideal
rare-event distribution $Q_{\mathcal R}$ satisfies
\begin{align}
    D_{\mathrm{tv}}(\widetilde Q,Q_{\mathcal R})
    &=
    \frac12
    \left[
    \sum_{x\in\mathcal R_\Delta}
    \left|
        \frac{Q(x)}{p_{\mathrm{rare}}+\widetilde q_\Delta}
        -
        \frac{Q(x)}{p_{\mathrm{rare}}}
    \right|
    +
    \sum_{x\in\widetilde{\mathcal S}_\Delta}
    \frac{Q(x)}{p_{\mathrm{rare}}+\widetilde q_\Delta}
    \right]
    \nonumber\\
    &=
    \frac{\widetilde q_\Delta}
    {p_{\mathrm{rare}}+\widetilde q_\Delta}
    \nonumber\\
    &\le
    \frac{q_\Delta}{p_{\mathrm{rare}}}
    =
    \zeta.
\end{align}
Including the failure probability $\delta$ of the threshold-classification
step gives total error at most $\zeta+\delta$. Taking
$\delta=O(\epsilon)$, the final total variation error is
$O(\epsilon+\zeta)$.

In total, the number of samples from $P$ is
\begin{align}
    O \left(
        \min\left\{
            \frac{1}{\Delta}\log\frac{N}{\epsilon},
            \frac{1}{\Delta^2}\log\frac{1}{\epsilon}
        \right\}
    \right),
\end{align}
and the number of samples from $Q$ is
\begin{align}
    O \left(
        \max\left\{
            \frac{1}{(p_1-p_2-q_\Delta)^2},
            \frac{1}{p_{\mathrm{rare}}}
        \right\}
    \right).
\end{align}
\end{proof}

Following this, one can immediately provide the classical algorithm for rare event sampling.

\begin{corollary}[Classical algorithm for rare event sampling]
    There exists a classical algorithm that solves \cref{prob.rare.specific} using
    $$
    \mathcal{O}\left(\min\left\{\frac{1}{\Delta}\log\left(\frac{N}{\epsilon}\right),\frac{1}{\Delta^2}\log\left(\frac{1}{\epsilon}\right)\right\}+\frac{1}{p_\mathrm{rare}}\right)
    $$
    samples from $P$.
\end{corollary}
\begin{proof}
The result follows from the construction in \cref{thm.alg.classical} specialized to the self-sampling case $Q=P$.
    In \cref{prob.rare.specific}, there is no ``Impossible'' branch, so the decision step distinguishing $p_{\mathrm{rare}}\ge p_1$ from $p_{\mathrm{rare}}\le p_2$ is unnecessary.

    Using the threshold-classification step from the proof of \cref{thm.alg.classical}, with probability at least $1-\delta$ one constructs a candidate set $\widehat R$ satisfying $\mathcal R_\Delta
        \subseteq
        \widehat R
        \subseteq
        \mathcal R_\Delta\cup\mathcal S_\Delta$
    using
    \begin{align}
        \mathcal{O}\left(
            \min\left\{
                \frac{1}{\Delta}\log\left(\frac{N}{\delta}\right),
                \frac{1}{\Delta^2}\log\left(\frac{1}{\delta}\right)
            \right\}
        \right)
    \end{align}
    samples from $P$.

    On this good event, $\widehat R=\mathcal R_\Delta\cup\widetilde{\mathcal S}_\Delta$ for some $\widetilde{\mathcal S}_\Delta\subseteq\mathcal S_\Delta$.
    Rejection sampling from $P$ using $\widehat R$ therefore accepts with probability $ P(\widehat R)=p_{\mathrm{rare}}+\widetilde p_\Delta$,
    where $\widetilde p_\Delta:=P(\widetilde{\mathcal S}_\Delta)\le p_\Delta$.
    Hence one accepted sample requires $\mathcal{O}(1/p_{\mathrm{rare}})$ samples from $P$.

    The same TVD calculation as in the proof of \cref{thm.alg.classical} gives
    \begin{align}
        D_{\mathrm{tv}}(\widetilde P_{\mathcal R},P_{\mathcal R})
        \le
        \frac{\widetilde p_\Delta}{p_{\mathrm{rare}}+\widetilde p_\Delta}
        \le
        \frac{p_\Delta}{p_{\mathrm{rare}}}
        =
        \zeta.
    \end{align}
    Including the failure probability of the candidate-set construction gives error at most $\zeta+\delta$.
    Taking $\delta=\mathcal{O}(\epsilon)$ yields total variation error $\mathcal{O}(\epsilon+\zeta)$ and gives the stated sample complexity.
\end{proof}

Now we prove the classical lower bound, where the proof is achieved by reduction from a standard
distribution distinguishing problem.

\begin{lemma}[Lower bound for distribution distinguishing \cite{BY02}]\label{lem.distinguish}
    For any two probability distributions $P_1$ and $P_2$ over the same sample space, any classical algorithm that distinguishes $P_1$ from $P_2$ requires $\Omega(1/d_H^2(P_1,P_2))$ samples, where 
    \begin{align}
        d_{H}(P_1,P_2)\coloneqq \sqrt{\frac{1}{2}\sum_{i}\left(\sqrt{P_1(i)}-\sqrt{P_2(i)}\right)^2}
    \end{align}
     is the Hellinger distance between $P_1$ and $P_2$.
\end{lemma}

\begin{theorem}[Classical lower bound in $\Delta$]\label{thm.lowerbound.classical}
    Any classical algorithm that solves \cref{prob.rare.general} or \cref{prob.rare.specific} with $0<\Delta<\frac{1}{4}$ requires $\Omega(1/\Delta)$ samples drawn from $P$.
\end{theorem}

\begin{proof}
    We prove the lower bound using a single pair of yes instances. 
    This avoids relying on the ``Impossible'' branch and therefore applies both to the general rare-event sampling problem and to the self-sampling version.

    Fix a constant output accuracy parameter $\eta<1/10$.
    Consider the following two probability distributions on the sample space $\{1,2,3\}$:
    \begin{align}
        P_0&=\left(\frac{\Delta}{2}, 2\Delta, 1-\frac{5\Delta}{2}\right),\\
        P_1&=\left(2\Delta, \frac{\Delta}{2}, 1-\frac{5\Delta}{2}\right).
    \end{align}
    For $0<\Delta<1/4$, both are valid probability distributions, and
    \begin{align}
        1-\frac{5\Delta}{2}>\frac{3\Delta}{2}.
    \end{align}
    Hence, for both distributions, the boundary region
    $\mathcal{S}_{\Delta}=\{x:\Delta\le P(x)\le 3\Delta/2\}$ is empty.

    We first compute the distance between the two input distributions. Their squared Hellinger distance is
    \begin{align}
        d_H^2(P_0,P_1)
        &=
        \frac{1}{2}
        \left[
        \left(\sqrt{\frac{\Delta}{2}}-\sqrt{2\Delta}\right)^2
        +
        \left(\sqrt{2\Delta}-\sqrt{\frac{\Delta}{2}}\right)^2
        \right]
        \nonumber\\
        &=
        \left(\sqrt{2\Delta}-\sqrt{\frac{\Delta}{2}}\right)^2
        =
        \frac{\Delta}{2}
        =
        \Theta(\Delta).
    \end{align}
    By \cref{lem.distinguish}, any classical algorithm that distinguishes $P_0$ from $P_1$ with constant success probability requires
    $\Omega(1/\Delta)$ samples.

    We now show that any rare-event sampling algorithm would distinguish these two distributions.
    In the general problem \cref{prob.rare.general}, choose the fixed distribution
    \begin{align}
        Q=\left(\frac{1}{2}, \frac{1}{2}, 0\right),
        \qquad
        p_1=\frac{1}{3},
        \qquad
        p_2=0.
    \end{align}
    Since the boundary region is empty, we have $q_{\Delta}=0$.
    The promise condition is satisfied for any sufficiently small constant precision parameter, because $p_1-p_2=1/3>0=q_{\Delta}$.
    If the unknown distribution is $P_0$, then the rare set is $R_{\Delta}(P_0)=\{1\}$,
    and hence the target rare-event distribution is the delta function $\delta_1$.
    Moreover,
    \begin{align}
        p_{\mathrm{rare}}
        =
        Q(R_{\Delta}(P_0))
        =
        Q(\{1\})
        =
        \frac{1}{2}
        \ge p_1.
    \end{align}
    If the unknown distribution is $P_1$, then the rare set is $R_{\Delta}(P_1)=\{2\}$,
    and hence the target rare-event distribution is the delta function $\delta_2$.
    Similarly,
    \begin{align}
        p_{\mathrm{rare}}
        =
        Q(R_{\Delta}(P_1))
        =
        Q(\{2\})
        =
        \frac{1}{2}
        \ge p_1.
    \end{align}
    Thus both inputs are yes instances of the general problem, but their correct output samplers are different.

    The same hard pair also applies to the self-sampling problem, where $Q=P$.
    Indeed, if $P=P_0$, then the rare set is again $\{1\}$ and the conditional rare-event distribution is $P_{0,R}=\delta_1$.
    If $P=P_1$, then the rare set is again $\{2\}$ and the conditional rare-event distribution is
    $P_{1,R}=\delta_2$.
    In both cases the rare-event mass is nonzero, i.e., $p_{\mathrm{rare}}=\Delta/2>0$.
    Therefore the pair $P_0,P_1$ consists of valid yes instances also for the sampling-only self-sampling problem.

    Now suppose that there exists a classical algorithm $\mathcal{A}$ that solves either the general problem above or the self-sampling problem using $S$ samples from the unknown distribution $P$.
    We use $\mathcal{A}$ to distinguish whether the unknown input is $P_0$ or $P_1$.
    Run $\mathcal{A}$ on the corresponding rare-event sampling instance, then draw one sample from the sampler produced by $\mathcal{A}$.
    If the sample is $1$, output $P_0$.
    If the sample is $2$, output $P_1$.
    In all other cases, output either distribution arbitrarily.

    Conditioned on $\mathcal{A}$ succeeding, its output distribution is within total variation distance at most $\eta$ from the correct rare-event distribution.
    Therefore, if $P=P_0$, the produced sampler outputs $1$ with probability at least $1-\eta$.
    If $P=P_1$, the produced sampler outputs $2$ with probability at least $1-\eta$.
    Since $\eta<1/10$ and $\mathcal{A}$ succeeds with constant probability, this gives a distinguisher for $P_0$ and $P_1$ with constant success probability.
    By the Hellinger-distance lower bound above, such a distinguisher requires
    \begin{align}
        S=\Omega\left(\frac{1}{\Delta}\right).
    \end{align}
    Therefore any classical algorithm solving the rare-event sampling problem requires $\Omega(1/\Delta)$ samples from $P$.
    The conclusion holds both for the general problem and for the sampling-only self-sampling problem.
\end{proof}

Note that the above proof naturally generalizes to the case where we can only sample from the rare distribution approximately. Using the hard instance in the proof of \cref{thm.lowerbound.classical}, $P_{\mathcal{R}}=(1,0,0)$. Performing the rare event sampling algorithm on distributions $P_1$ and $P_2$, we output failure and distribution $(1,0,0)$, respectively.
Even if we can only obtain the final output to precision $1/3$ in total variation distance, we can still distinguish $P_1$ and $P_2$ by taking a constant number of samples from the output distribution.

\subsection{Quantum algorithm \label{app.alg.quantum}}

\begin{theorem}[Quantum algorithm for rare event sampling]\label{thm.alg.rare.pure}
    There exists a quantum algorithm that solves \cref{prob.rare.general} using
    \begin{align}
        \mathcal{O}\left(  \max\left\{\frac{1}{\sqrt{p_1}-\sqrt{p_2}}, \frac{1}{\sqrt{p_{\mathrm{rare}}}} \right\}\frac{1}{\sqrt{\Delta}} \log\left(\frac{1}{\epsilon}\right)\right)
    \end{align}
    queries to $U_P$ and $\mathcal{O}\left( \max\left\{\frac{1}{\sqrt{p_1}-\sqrt{p_2}}, \frac{1}{\sqrt{p_{\mathrm{rare}}}} \right\}\right)$ queries to $U_Q$ with $\mathcal{O}(n)$ ancilla qubits.
    Here, $U_P$ and $U_Q$ denote the quantum samplers for $P$ and $Q$, respectively.
\end{theorem}

\begin{proof}
    By \cref{thm.amplitudes.encoding}, we can use $\mathcal{O}(1)$ times of controlled-$U_P$ and controlled-$U_P^\dagger$ to construct an $(1,n+2,0)$-encoding $U_A$ of a diagonal matrix $A=\diag(\{\sqrt{p(x_{i})}\})$.

    The next step is to apply the rectangular function to the diagonal matrix $A$ of the encoding $U_A$, which approximately constructs a projection $\Pi_\mathcal{R}$ into the space filled with rare events.
    Note that in the quantum setting, we operate on the square roots of the probabilities rather than the probabilities themselves.
    To ensure the ambiguity range is contained in $\mathcal{S}_{\Delta}$, for quantum we choose the ambiguity range as $0.2\sqrt{\Delta}$ such that
    \begin{align}
        (\sqrt{\Delta}+0.2\sqrt{\Delta})^2 = 1.44\Delta<1.5\Delta.
    \end{align}
    \cref{approx.rectangle} guarantees that the rectangular function can be approximated to error $\epsilon$ by an even polynomial $P$ with degree $\mathcal{O}(\frac{1}{\sqrt{\Delta}}\log(1/\epsilon))$, with $t=(\sqrt{\Delta}+\xi)/2=1.1\sqrt{\Delta}$ and $\Gamma=\xi/2=0.1\sqrt{\Delta}$.
    \cref{theorem.qsvt} claims that such an even polynomial can be applied to the diagonal matrix $A$ by using $\mathcal{O}(\frac{1}{\sqrt{\Delta}}\log(1/\epsilon))$ queries to controlled-$U_A$ and controlled-$U_A^\dagger$.
    With these, we can approximately construct a block encoding of the operator
    \begin{equation}
        \Pi_\mathcal{R} = \sum_{x_{i}\in \mathcal{R}}\ket{x_{i}}\bra{x_{i}}.
    \end{equation}

    Then, we apply the block encoding $U_{\Pi_\mathcal{R}}$ to the initial state prepared by $U_Q$, that is,
    \begin{align}\label{eq.state.prop.naa}
        U_{\Pi_\mathcal{R}}(I\otimes U_{Q})\ket{0}\ket{0}=
        U_{\Pi_\mathcal{R}} \ket{0}\sum_{x_{i}} \sqrt{q(x_{i})}\ket{x_{i}}\approx \ket{0}\sum_{x_{i}\in \mathcal{R}}\sqrt{q(x_{i})} \ket{x_{i}} +\ket{1}\ket{\widetilde{\perp}},
    \end{align}
    where $\ket{\widetilde{\perp}}$ is an arbitrary unnormalized state.
    For simplicity, write $\ket{\phi}=U_Q\ket{0}$.
    The probability of measuring the ancilla qubits  as $\ket{0}$ is
    \begin{equation}
        \|\tilde{\Pi}_\mathcal{R}\ket{\phi} \|^2 \geq \|\Pi_\mathcal{R}\ket{\phi} \|^2 (1-\epsilon)^2 = p_{\mathrm{rare}}(1-\epsilon)^2,
    \end{equation}
    where $\tilde{\Pi}_\mathcal{R}$ is the actual constructed approximated projector.
    We also have
    \begin{align}
        \|\tilde{\Pi}_\mathcal{R}\ket{\phi} \|^2\leq \|\Pi_\mathcal{R}\ket{\phi} \|^2+q_{\Delta}=p_{\mathrm{rare}}+q_{\Delta}+\epsilon^2.
    \end{align}

    To decide whether we are in the first or the second case as required in \cref{prob.rare.general},
    it suffices to achieve this by measuring $\|\tilde{\Pi_\mathcal{R}}\ket{\phi} \|^2$ with amplitude estimation \cite{brassard2002quantum}. To separate these two cases, we need to make sure
    \begin{align}
        p_1(1-\epsilon)^2\geq p_2+q_{\Delta}+\epsilon^2.
    \end{align}
    Note that
    \begin{align}
        p_1(1-\epsilon)^2=p_1(1-2\epsilon+\epsilon^2)\geq p_1(1-2\epsilon),
    \end{align}
    it suffices to take $p_1-p_2>2\epsilon+q_{\Delta}\geq 2\epsilon p_1+q_{\Delta}$.
    One can further rescale $2\epsilon$ to $\epsilon$.

    Note that \cref{eq.state.prop.naa} is also a state preparation unitary. One can ignore the second register and only consider the first register, which can be written as
    \begin{align}
        U_{\mathrm{prep}}\ket{0}=\|\tilde{P}_{\mathrm{rare}}\ket{\phi}\|\ket{0}+\sqrt{1-\|\tilde{P}_{\mathrm{rare}}\ket{\phi}\|^2}\ket{1}.
    \end{align}
    By using \cref{thm.amplitudes.encoding} and implementing the approximate rectangular function $f(x-\frac{\sqrt{p_1}+\sqrt{p_2}}{2})$, with $\Gamma'= (\sqrt{p_1}-\sqrt{p_2})/2$, we can boost the gap between $p_1$ and $p_2$ to $\Omega(1)$. Then we can use amplitude amplification \cite{brassard2002quantum} with precision $\mathcal{O}(1)$ to distinguish the first case and the second case with probability at least $2/3$.
    To distinguish whether we are in the first case or in the second case uses $\mathcal{O}(\frac{1}{\sqrt{p_1}-\sqrt{p_2}}\frac{1}{\sqrt{\Delta}}\log\frac{1}{\epsilon})$ queries to $U_P$ and $\mathcal{O}(\frac{1}{\sqrt{p_1}-\sqrt{p_2}})$ queries to $U_Q$.

    Now we consider if we are in \textbf{the first case}.

    By further using the amplitude amplification $\mathcal{O}(\frac{1}{\sqrt{p_{\mathrm{rare}}}})$ times \cite{brassard2002quantum}, one can prepare the target state
    \begin{align}
        \frac{1}{\sqrt{p_{\mathrm{rare}}}} \sum_{x_i\in \mathcal{R}} \sqrt{Q(x_i)}\ket{x_i}.
    \end{align}
    Combining the QSVT procedure, the overall queries to $U_P$ is $\mathcal{O}(\frac{1}{\sqrt{\Delta p_{\mathrm{rare}}}}\log\frac{1}{\epsilon})$ and queries to $U_Q$ is $\mathcal{O}(\frac{1}{\sqrt{p_{\mathrm{rare}}}})$.

Let $    \tilde r
    :=
    \left\|\widetilde\Pi_{\mathcal R}\ket{\phi}\right\|^2
    =
    \sum_i
    \left|\bra{x_i}\widetilde\Pi_{\mathcal R}\ket{\phi}\right|^2$
be the actual post-selection probability. The resulting post-selected distribution is
\begin{equation}
    \widetilde{Q}_{\mathcal{R}}(x_i)
    =
    \frac{
        \left|\bra{x_i}\widetilde\Pi_{\mathcal R}\ket{\phi}\right|^2
    }{\tilde r}.
\end{equation}

Let $\mathcal{S}_{\mathrm{un}}
    :=
    \Omega\setminus(\mathcal{R}\cup\mathcal{S}_{\Delta})$
denote the non-rare region. By the polynomial approximation guarantee, for $x_i\in\mathcal R$,
\begin{equation}
    (1-\epsilon)^2Q(x_i)
    \le
    \left|\bra{x_i}\widetilde\Pi_{\mathcal R}\ket{\phi}\right|^2
    \le
    Q(x_i).
\end{equation}
For $x_i\in\mathcal S_\Delta$, we use the trivial bound
\begin{equation}
    \left|\bra{x_i}\widetilde\Pi_{\mathcal R}\ket{\phi}\right|^2
    \le
    Q(x_i),
\end{equation}
and for $x_i\in\mathcal S_{\mathrm{un}}$,
\begin{equation}
    \left|\bra{x_i}\widetilde\Pi_{\mathcal R}\ket{\phi}\right|^2
    \le
    \epsilon^2 Q(x_i).
\end{equation}
Therefore,
\begin{equation}
    (1-\epsilon)^2p_{\mathrm{rare}}
    \le
    \tilde r
    \le
    p_{\mathrm{rare}}+q_\Delta+\epsilon^2.
\end{equation}
Moreover,
\begin{equation}
    |\tilde r-p_{\mathrm{rare}}|
    \le
    \left (1-(1-\epsilon)^2\right)p_{\mathrm{rare}}
    +q_\Delta+\epsilon^2 .
\end{equation}

For the rare region, we first note that
\begin{align}
&\sum_{x_i\in\mathcal R}
\left|
    \frac{Q(x_i)}{p_{\mathrm{rare}}}
    -
    \frac{
        \left|\bra{x_i}\widetilde\Pi_{\mathcal R}\ket{\phi}\right|^2
    }{\tilde r}
\right|
\nonumber\\
\le&
\sum_{x_i\in\mathcal R}
Q(x_i)
\left|
    \frac{1}{p_{\mathrm{rare}}}
    -
    \frac{1}{\tilde r}
\right|
+
\sum_{x_i\in\mathcal R}
\frac{
    Q(x_i)
    -
    \left|\bra{x_i}\widetilde\Pi_{\mathcal R}\ket{\phi}\right|^2
}{\tilde r}.
\end{align}
The second term is nonnegative because
$\left|\bra{x_i}\widetilde\Pi_{\mathcal R}\ket{\phi}\right|^2\le Q(x_i)$
for $x_i\in\mathcal R$.

We evaluate the total variation distance between the target distribution
$Q_{\mathcal R}$ and the actual distribution $\widetilde Q_{\mathcal R}$:
\begin{align}
D_{\mathrm{tv}}(Q_{\mathcal R},\widetilde Q_{\mathcal R})
=&
\frac12
\sum_{x_i}
\left|
Q_{\mathcal R}(x_i)-\widetilde Q_{\mathcal R}(x_i)
\right|
\nonumber\\
\le&
\frac12
\left\{
\sum_{x_i\in\mathcal R}
\left|
    \frac{Q(x_i)}{p_{\mathrm{rare}}}
    -
    \frac{
        \left|\bra{x_i}\widetilde\Pi_{\mathcal R}\ket{\phi}\right|^2
    }{\tilde r}
\right|
+
\sum_{x_i\in\mathcal S_\Delta}
\frac{Q(x_i)}{\tilde r}
+
\sum_{x_i\in\mathcal S_{\mathrm{un}}}
\frac{\epsilon^2Q(x_i)}{\tilde r}
\right\}
\nonumber\\
\le&
\frac12
\left\{
\sum_{x_i\in\mathcal R}
Q(x_i)
\left|
    \frac1{p_{\mathrm{rare}}}
    -
    \frac1{\tilde r}
\right|
+
\sum_{x_i\in\mathcal R}
\frac{
    Q(x_i)
    -
    \left|\bra{x_i}\widetilde\Pi_{\mathcal R}\ket{\phi}\right|^2
}{\tilde r}
+
\frac{q_\Delta}{\tilde r}
+
\frac{\epsilon^2}{\tilde r}
\right\}
\nonumber\\
\le&
\frac12
\left\{
\frac{|\tilde r-p_{\mathrm{rare}}|}{\tilde r}
+
\frac{\left(1-(1-\epsilon)^2\right)p_{\mathrm{rare}}}{\tilde r}
+
\frac{q_\Delta}{\tilde r}
+
\frac{\epsilon^2}{\tilde r}
\right\}
\nonumber\\
\le&
\frac12
\left\{
\frac{
    2\left(1-(1-\epsilon)^2\right)p_{\mathrm{rare}}
    +2q_\Delta
    +2\epsilon^2
}{
    \tilde r
}
\right\}
\nonumber\\
\le&
\frac{
    \left(1-(1-\epsilon)^2\right)p_{\mathrm{rare}}
    +q_\Delta
    +\epsilon^2
}{
    (1-\epsilon)^2p_{\mathrm{rare}}
}
\nonumber\\
=&
\frac{2\epsilon-\epsilon^2}{(1-\epsilon)^2}
+
\frac{\zeta}{(1-\epsilon)^2}
+
\frac{\epsilon^2}{(1-\epsilon)^2p_{\mathrm{rare}}},
\end{align}
where $\zeta=q_\Delta/p_{\mathrm{rare}}$.

We choose the polynomial approximation error such that $\epsilon\le p_1$.
In the yes case of \cref{prob.rare.general}, $p_{\mathrm{rare}}\ge p_1$. Hence
\begin{equation}
    \frac{\epsilon^2}{p_{\mathrm{rare}}}
    \le
    \frac{\epsilon^2}{p_1}
    \le
    \epsilon.
\end{equation}
Therefore, assuming $\epsilon\leq 1/2$ without loss of generality,
\begin{equation}
    D_{\mathrm{tv}}(Q_{\mathcal R},\widetilde Q_{\mathcal R})
    =
    \mathcal O(\epsilon+\zeta).
\end{equation}

    If we are in \textbf{the second case}, we can directly output ``impossible".
    Combining these two cases, one needs
    \begin{align}
        \mathcal{O}\left(\max\left\{\frac{1}{\sqrt{p_1}-\sqrt{p_2}}, \frac{1}{\sqrt{p_{\mathrm{rare}}}}\right\}\frac{1}{\sqrt{\Delta}}\log\frac{1}{\epsilon}\right)
    \end{align}
    queries to $U_P$, and $\mathcal{O}\left(\max\left\{\frac{1}{\sqrt{p_1}-\sqrt{p_2}}, \frac{1}{\sqrt{p_{\mathrm{rare}}}}\right\}\right)$ queries to $U_Q$.
\end{proof}

\begin{figure}
\vspace*{0.5em}
\begin{tcolorbox}[enhanced]
\centering
\textbf{Quantum algorithm for rare event sampling}
    \vspace{0.5em}

    \textbf{Inputs:}

    \vspace{0.25em}
    \begin{tabular}{@{} l @{\hspace{0.8em}--\hspace{0.8em}} p{0.56\linewidth} @{}}
        c-$U_P/U_P^\dagger$
        & $U_P$ generates the quantum sample state $\ket{P}$. \\
        $\Delta$
        & Threshold for rare events. \\
        $\epsilon$
        & Error parameter.
    \end{tabular}

    \vspace{0.75em}

    \textbf{Outputs:}

    \vspace{0.25em}
    \begin{tabular}{@{} l @{\hspace{0.8em}--\hspace{0.8em}} p{0.56\linewidth} @{}}
        $\ket{P_{\mathcal{R}}}$
        & A quantum rare-event sample state that superposes all rare events.
    \end{tabular}

    \vspace{0.75em}

    \textbf{Algorithm:}

    \vspace{0.25em}
    \begin{enumerate}[leftmargin=*]
        \item By \cref{thm.amplitudes.encoding}, use c-$U_P/U_P^\dagger$ a constant number of times to construct the diagonal block encoding of amplitudes $U_A$.

        \item Implement a polynomial approximation of the rectangle function $f(x)$ via \cref{theorem.qsvt}, using
        $\mathcal{O}(1/\sqrt{\Delta}\log(1/\epsilon))$ applications of $U_A$.

        \item Feed the quantum sample state $\ket{P}$ into the constructed quantum circuit, and use amplitude amplification to boost the success probability and prepare the target state $\ket{P_{\mathcal{R}}}$. This further uses $\mathcal{O}(1/\sqrt{p_{\mathrm{rare}}})$ applications of the constructed quantum circuit. In total, the algorithm uses
        $\mathcal{O}(1/\sqrt{\Delta p_{\mathrm{rare}}}\log(1/\epsilon))$ applications of c-$U_P/U_P^\dagger$.
    \end{enumerate}
\end{tcolorbox}
\caption{Description of the quantum rare-event sampling algorithm, with the special case $P=Q$.}
\label{fig.algorithm.step}
\end{figure}

Following this, we can immediately write the case for rare event sampling.

\begin{corollary}[Quantum algorithm for rare event sampling]
    There exists a quantum algorithm that solves \cref{prob.rare.specific} using
    \begin{align}
        \mathcal{O}\left( \frac{1}{\sqrt{p_{\mathrm{rare}}\Delta}} \log\left(\frac{1}{\epsilon}\right)\right)
    \end{align}
    queries to $U_P$ and $U_P^\dagger$ with $\mathcal{O}(n)$ ancilla qubits.
    Here, $U_P$ denote the quantum samplers for $P$.
\end{corollary}
\begin{proof}
    This is the self-sampling case $Q=P$ of \cref{thm.alg.rare.pure}, with the ``Impossible'' decision branch omitted.
    In \cref{prob.rare.specific}, the rare-event mass $p_{\mathrm{rare}}$ is assumed to be nonzero, and the goal is only to prepare a sampler close to the conditional rare-event distribution.

    The thresholding construction from \cref{thm.alg.rare.pure} uses
    \begin{align}
        \mathcal{O}\left(
            \frac{1}{\sqrt{\Delta}}
            \log\left(\frac{1}{\epsilon}\right)
        \right)
    \end{align}
    queries to $U_P$ and $U_P^\dagger$ to implement the approximate rare-event projector.
    Applying this projector to the quantum sample state $U_P\ket{0}$ succeeds with probability $\mathcal{O}(p_{\mathrm{rare}})$.
    Amplitude amplification therefore adds a factor $\mathcal{O}(1/\sqrt{p_{\mathrm{rare}}})$.
    Hence the total query complexity is
    \begin{align}
        \mathcal{O}\left(
            \frac{1}{\sqrt{p_{\mathrm{rare}}\Delta}}
            \log\left(\frac{1}{\epsilon}\right)
        \right).
    \end{align}
    The total variation error bound is the same as in the proof of \cref{thm.alg.rare.pure}, with $Q=P$, giving $\mathcal{O}(\epsilon+\zeta)$.
\end{proof}

In the following, we provide the lower bound proof for the quantum case.

\begin{theorem}[Quantum lower bound in $\Delta$]\label{thm.lowerbound.quantum}
    For any $0<\Delta<1/4$, any quantum algorithm that solves
    \cref{prob.rare.general} with quantum samplers $U_P,U_Q$,
    parameters $p_1=2/3$, $p_2=1/3$, and $0<\epsilon<1/3$,
    requires
    \[
        \Omega \left(\frac{1}{\sqrt{\Delta}}\right)
    \]
    queries to $U_P$ and $U_P^\dagger$, even when $Q$ is fixed and known.
\end{theorem}

\begin{proof}
    We prove the lower bound by reducing a simple quantum-sampler
    distinguishing task to rare-event sampling.
    Consider two probability distributions $P_0$ and $P_1$ over the sample space
    $\Omega=\{0,1\}$: $P_0(1)=\Delta/2$ and $P_0(0)=1-\Delta/2$; $P_1(1)=2\Delta$ and $P_1(0)=1-2\Delta$.
    Since $0<\Delta<1/4$, these are valid probability distributions.
    Also we have $P_0(1)<\Delta$, $P_1(1)>\frac{3\Delta}{2}$,
    and both $P_0(0)$ and $P_1(0)$ are larger than $3\Delta/2$.
    Therefore the boundary region $\mathcal{S}_{\Delta}$
    is empty for both instances.

    For $b\in\{0,1\}$, let the corresponding quantum samplers be the one-qubit rotations
    \begin{align}
        U_b\ket{0}
        =
        \sqrt{1-P_b(1)}\ket{0}
        +
        \sqrt{P_b(1)}\ket{1}.
    \end{align}
    Equivalently, write
    \begin{align}
        U_b\ket{0}
        =
        \cos\theta_b\ket{0}
        +
        \sin\theta_b\ket{1},
    \end{align}
    where $\theta_0=\arcsin\sqrt{\frac{\Delta}{2}}$ and $\theta_1=\arcsin\sqrt{2\Delta}$. 
    For $0<\Delta<1/4$, this gives $|\theta_1-\theta_0|=\Theta(\sqrt{\Delta})$.
    We first show that distinguishing these two quantum samplers requires
    $\Omega(1/\sqrt{\Delta})$ queries. Choose the canonical rotation
    implementation
\begin{align}
            U_b
        =
        \begin{pmatrix}
            \cos\theta_b & -\sin\theta_b\\
            \sin\theta_b & \cos\theta_b
        \end{pmatrix}.
\end{align}
    Then
    \begin{align}
        \|U_0-U_1\|
        =
        2\left|\sin\frac{\theta_1-\theta_0}{2}\right|
        =
        O(|\theta_1-\theta_0|)
        =
        O(\sqrt{\Delta}).
    \end{align}
    The same bound holds for the inverse unitaries $\|U_0^\dagger-U_1^\dagger\|
        =
        O(\sqrt{\Delta})$.

    Consider an arbitrary quantum algorithm making $T$ queries to the
    unknown sampler and its inverse, interleaved with arbitrary known
    unitaries. Let $\ket{\Psi_0^{(T)}}$ and $\ket{\Psi_1^{(T)}}$ denote
    the final pure states of the algorithm when the oracle is $U_0$ and
    $U_1$, respectively. By the standard hybrid argument for quantum
    query algorithms~\cite{bennett1997strengths}, replacing the oracle
    calls one by one gives
    \begin{align}
                \left\|
            \ket{\Psi_0^{(T)}}-\ket{\Psi_1^{(T)}}
        \right\|
        \le
        T\cdot
        \max\left\{
            \|U_0-U_1\|,
            \|U_0^\dagger-U_1^\dagger\|
        \right\}=\mathcal{O}(T\sqrt{\Delta}).
    \end{align}

    On the other hand, if the algorithm distinguishes the two cases with
    success probability at least $2/3$, then the final states must have
    constant trace distance. Indeed, by the Helstrom--Holevo theorem~\cite{helstrom1967detection,holevo1973statistical,watrous2018theory}, the
    optimal success probability for distinguishing two equally likely
    states $\rho_0,\rho_1$ is
    \begin{align}
         p_{\mathrm{succ}}
        =
        \frac12+\frac14\|\rho_0-\rho_1\|_1.
    \end{align}
    Thus $p_{\mathrm{succ}}\ge 2/3$ implies $\|\rho_0-\rho_1\|_1\geq 2/3$.
    For pure states $\rho_b=\ket{\Psi_b^{(T)}} \bra{\Psi_b^{(T)}}$, we have
    \begin{align}
         \frac12\|\rho_0-\rho_1\|_1
    =
    \sqrt{1-\left|\braket{\Psi_0^{(T)}}{\Psi_1^{(T)}}\right|^2}
    \le
    \left\|
        \ket{\Psi_0^{(T)}}-\ket{\Psi_1^{(T)}}
    \right\|.
    \end{align}
Therefore, $\|\rho_0-\rho_1\|_1\ge 2/3$ implies
\begin{align}
        \left\|
        \ket{\Psi_0^{(T)}}-\ket{\Psi_1^{(T)}}
    \right\|
    \ge
    \frac13
    =
    \Omega(1).
\end{align}

    Combining this with the hybrid bound yields $O(T\sqrt{\Delta})=\Omega(1)$,
    and therefore
    \begin{align}
        T=\Omega \left(\frac{1}{\sqrt{\Delta}}\right).
    \end{align}

    We now show that any rare-event sampling algorithm would distinguish
    $U_0$ from $U_1$. Set $Q=(0,1)$, $p_1=2/3$, $p_2=1/3$, and $0<\epsilon<1/3$.
    Since the boundary region is empty for both instances, we have
    $q_\Delta=0$. Hence the promise condition
    $p_1-p_2>\epsilon+q_\Delta$
    is satisfied.

    Suppose there exists a quantum algorithm $\mathcal A$ that solves
    \cref{prob.rare.general} using $S$ queries to $U_P$ and
    $U_P^\dagger$. We use $\mathcal A$ to distinguish whether the
    unknown sampler is $U_0$ or $U_1$. Run $\mathcal A$ with quantum
    sampling access to the unknown $P$, and with the fixed known sampler
    $U_Q$ for $Q=(0,1)$. If $\mathcal A$ outputs ``Impossible'', we
    output $P_1$. Otherwise, we output $P_0$.

    If $P=P_0$, then outcome $1$ is rare because $P_0(1)=\frac{\Delta}{2}<\Delta$.
    Since $Q=(0,1)$, the $Q$-mass of the rare set is $p_{\mathrm{rare}}=Q(\{1\})=1\ge p_1$.
    Therefore $\mathcal A$ must construct a sampler for the rare-event
    distribution with probability at least $2/3$.

    If $P=P_1$, then no outcome is rare: $P_1(1)=2\Delta>3\Delta/2$,
    while $P_1(0)>3\Delta/2$. Hence $ p_{\mathrm{rare}}
        =
        Q(\emptyset)
        =
        0
        \le
        p_2$.
    By the definition of \cref{prob.rare.general}, $\mathcal A$ must
    output ``Impossible'' with probability at least $2/3$.

    Thus $\mathcal A$ distinguishes the two quantum samplers $U_0$ and
    $U_1$ with constant success probability using $S$ queries to
    $U_P$ and $U_P^\dagger$. Since this distinguishing task requires
    $\Omega(1/\sqrt{\Delta})$ queries, we conclude that
    \[
        S=\Omega \left(\frac{1}{\sqrt{\Delta}}\right).
    \]
\end{proof}

\begin{corollary}[Quantum lower bound for self-sampling rare-event sampling]
Fix a constant $\alpha>0$ in the definition
$S_\Delta=\{x:\Delta<P(x)\le (1+\alpha)\Delta\}$.
For all sufficiently small $\Delta>0$, any quantum algorithm that solves \cref{prob.rare.specific} with constant success probability and constant accuracy requires
\[
    \Omega\!\left(\frac{1}{\sqrt{\Delta}}\right)
\]
queries to $U_P$ and $U_P^\dagger$. This holds even for distributions with empty ambiguous region.
\end{corollary}

\begin{proof}
Choose any constant $c>1+\alpha$, and assume
$\Delta<1/(c+3/2+\alpha)$. Consider two distributions on
$\Omega=\{1,2,3\}$:

\begin{align}
    P_0
    &=
    \left(
        \frac{\Delta}{2},\, c\Delta,\,
        1-\left(c+\frac12\right)\Delta
    \right),\\
    P_1
    &=
    \left(
        c\Delta,\, \frac{\Delta}{2},\,
        1-\left(c+\frac12\right)\Delta
    \right).
\end{align}
For both distributions, the third probability is larger than
$(1+\alpha)\Delta$, and $c\Delta>(1+\alpha)\Delta$. Hence the ambiguous
region is empty. Under $P_0$, the rare set is $\{1\}$ and the
conditional rare-event distribution is $\delta_1$. Under $P_1$, the
rare set is $\{2\}$ and the conditional rare-event distribution is
$\delta_2$.

The squared Hellinger distance between the two input distributions is
\begin{align}
d_H^2(P_0,P_1)
&=
\frac12
\left[
\left(\sqrt{\frac{\Delta}{2}}-\sqrt{c\Delta}\right)^2
+
\left(\sqrt{c\Delta}-\sqrt{\frac{\Delta}{2}}\right)^2
\right]  \\
&=
\left(\sqrt c-\frac{1}{\sqrt2}\right)^2\Delta
=
\Theta(\Delta).
\end{align}
By the quantum distribution-distinguishing lower bound of Ref.~\cite{belovs:lipics.esa.2019.16},
distinguishing quantum samplers for $P_0$ and $P_1$ requires
\begin{align}
    \Omega\left(\frac{1}{d_H(P_0,P_1)}\right)
    =
    \Omega\left(\frac{1}{\sqrt{\Delta}}\right)
\end{align}
queries.

It remains to reduce this distinguishing task to rare-event sampling.
Suppose there were an algorithm $\mathcal A$ solving \cref{prob.rare.specific} with
$o(1/\sqrt{\Delta})$ queries. Given an unknown sampler for either
$P_0$ or $P_1$, run $\mathcal A$ and draw one sample from the
sampler it outputs. If the sample is $1$, output $P_0$; if the
sample is $2$, output $P_1$; otherwise output arbitrarily.
Conditioned on the success of $\mathcal A$, the output distribution is
within constant total variation distance of either $\delta_1$ or $\delta_2$.
Thus this procedure distinguishes $P_0$ from $P_1$ with constant
success probability. This contradicts the Hellinger-distance lower bound
above. Therefore any quantum algorithm solving \cref{prob.rare.specific} requires
$\Omega(1/\sqrt{\Delta})$ queries.
\end{proof}

\begin{figure}
    \includegraphics[width=\linewidth]{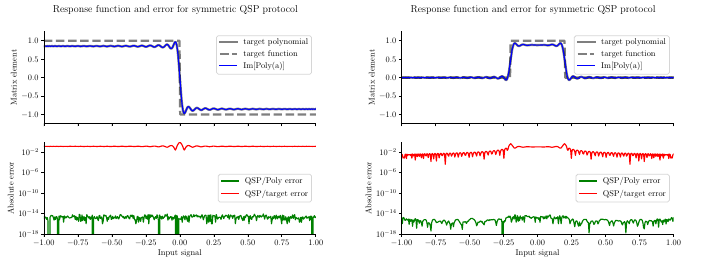}
    \caption{Polynomial approximation for the sign function (left) and rectangle function (right).}
    \label{fig.qsp.approx}
\end{figure}

\section{Discussion}

\subsection{Polynomial approximation}

Here, we provide some discussions of the polynomial approximation used in the quantum algorithm.
In this paper, we focus on the stochastic process, which implies that all quantum amplitudes (corresponding to the square root of the classical distribution) are real and positive.
To achieve the rare event-related tasks, we need a good polynomial approximation for the function
\begin{align}
    f(x)=
    \begin{cases}
        1 & \text{for all } x \in [0, \sqrt\Delta], \\
        0 & \text{for all } x \in (\sqrt\Delta, 1],
    \end{cases}
\end{align}
and $|f(x)|\leq 1$ for $x\in[-1,0)$.
The shifted Heaviside function is a good candidate.
There are two ways for approximating the Heaviside function: one is achieved by the error function \cite{low2017optimal}, which is an approximation of the sign function, and another is achieved by taking the integral of the filtering function \cite{lin2020nearoptimalground}, which is an approximation of the delta function.
Note that for the sign function case, one needs to add a constant $1$.
Asymptotically, the sign function approximation is better than the second method.

However, from practical performance, we notice that directly approximating the rectangle function  
\begin{align}
    g(x)=
    \begin{cases}
        1 & \text{for all } x \in [-\sqrt\Delta, \sqrt\Delta], \\
        0 & \text{for all } x \in [-1, -\sqrt\Delta) \cup (\sqrt\Delta, 1],
    \end{cases}
\end{align}
via optimization-based angle finding methods works better, which is somewhat surprising, as one usually constructs the rectangle function by taking linear combinations of the Heaviside function in theoretical proofs.
The reason is that in practice, we use an optimization-based method to find the polynomial approximation.
Especially for the quasi-Newton method \cite{dong2021efficient}, which is currently the state-of-the-art method to find many polynomial approximations, one can see from \cref{fig.qsp.approx} that the overall approximation for the rectangle function is much better than the sign function.
This difference in approximation error leads to the result that when we use the sign function in the numerical experiment, it requires a much higher polynomial degree to compress the largest amplitudes than the rectangle function case.

Further, the rectangle function needs fewer ancilla qubits than the sign function when we implement it via QSVT in the block-encoding setting.
The reason is that to implement the shifted sign function, one needs to make a shift and take a linear combination with the identity, as the shifted sign function is neither odd nor even and cannot be directly implementable via QSVT phase angles. This requires two ancilla qubits, and the coefficient from LCU will make amplitudes smaller.
For the rectangle function, as we can directly find the polynomial approximation in practice, it does not require additional ancilla qubits.

Therefore, in this paper, we choose to use the rectangle function.
The above discussions do not affect much for the asymptotic results, which are described in the theorems.
However, for practical considerations like the numerical experiment, these do make a difference.

\subsection{Related works}

The rare event sampling problem can be decomposed into two distinct tasks: first, identifying the rare events, and second, sampling from them. Our quantum algorithm achieves efficiency by performing these steps coherently, thereby avoiding the need to explicitly read out or measure the rare events.
In contrast, relying on the standard amplitude estimation algorithm \cite{brassard2002quantum} would incur a linear overhead in the size of the sample space, as it necessitates estimating the probability of each event individually. Even utilizing multidimensional quantum amplitude estimation \cite{Apeldoorn2021multidimensional} proves suboptimal, as it requires $\mathcal{O}(1/\Delta)$ queries to estimate all probabilities to the necessary precision of $\mathcal{O}(\Delta)$.
Further, to achieve a quadratic speedup with respect to $p_{\mathrm{rare}}$ during the sampling phase, one must effectively implement a projector onto the rare event subspace. Without our coherent approach, constructing such a projector would require complex quantum arithmetic circuits.
However, the trade-off for this coherent implementation is that our complexity scales with the product of $\mathcal{O}(1/\sqrt{\Delta})$ and $\mathcal{O}(1/\sqrt{p_{\mathrm{rare}}})$. In the worst-case scenario where there exists only a single rare event (that is, $p_{\mathrm{rare}}=\mathcal{O}(\Delta)$), our complexity converges to that of multidimensional quantum amplitude estimation.

Related problems include quantum minimum and maximum finding. The original quantum minimum finding algorithm, described in Ref.~\cite{durr1999quantumalgorithmfindingminimum}, is based on Grover's search algorithm. This approach has been generalized to the two-phase quantum search in Ref.~\cite{Apeldoorn2019Improvements}.
These methods assume that one can prepare the state
\begin{align}
    \frac{1}{\sqrt{|\mathcal{X}^L|}}\sum_{x\in \mathcal{X}^L} \ket{x}\ket{p(x)},
\end{align}
and one has access to an oracle $O$ such that
\begin{align}
    O\ket{p(x)}=
    \begin{cases}
        -\ket{p(x)} & \text{if } p(x)\leq \Delta, \\
        \ket{p(x)}  & \text{otherwise},
    \end{cases}
\end{align}
where $\ket{p(x)}$ corresponds to the binary representation of $p(x)$.
To prepare the state, a standard way is to assume another oracle
\begin{align}
    O_x:\ket{x}\ket{0}\rightarrow \ket{x}\ket{p(x)},
\end{align}
and implement this oracle onto the uniform superposition state.
Constructing these oracles may pose significant implementation challenges.
Similar limitations apply to the quantum maximum finding algorithm \cite{ahuja1999quantumalgorithmfindingmaximum}. While recent work \cite{rattew2023nonlinear} addresses maximum finding using a state preparation input model, their approach relies on restrictive assumptions regarding the largest amplitude and spectral gap (the difference between the largest and second-largest amplitudes).

Our approach draws conceptual inspiration from quantum ground state preparation (or low-energy state preparation) \cite{lin2020nearoptimalground, dong2022ground, buhrman2025beating}. In this analogy, our target output corresponds to a superposition of low-energy eigenstates, defined as those events with probability below the threshold $\Delta$. To utilize this insight, however, one must construct an effective Hamiltonian derived from the given state preparation unitary, which is highly nontrivial.
Furthermore, we explicitly prove the optimality of our quantum algorithm and demonstrate a separation between classical and quantum complexities. Our proof technique for these bounds is distinct from the method in Ref.~\cite{lin2020nearoptimalground}.

\section{Phase transition of quantum advantages in power-law tails}
\label{app:power_law_scaling}

In this section, we provide a rigorous discussion to determine in a concrete model family, when the rare-event mass $p_{\mathrm{rare}}$ stays macroscopic and when it vanishes. This is the key quantity controlling whether the quantum algorithm attains its ideal $\Theta(\Delta^{-1/2})$ scaling, or instead suffers an additional penalty from amplitude amplification.

Throughout this section, we fix $P=Q$, $\epsilon=\Theta(1)$ and focus only on the asymptotic dependence on $\Delta$. Under this convention, the result as stated in \cref{thm.alg.classical}, \cref{thm.alg.rare.pure} simplifies to
\begin{align}
\mathcal{S}_C(\Delta)
&=
\mathcal{O}\left(
\min\left\{
\frac{1}{\Delta}\log N,
\frac{1}{\Delta^2}
\right\}
+
\frac{1}{p_{\mathrm{rare}}}
\right),\\
\mathcal{S}_Q(\Delta)
&=
\mathcal{O}\left(
\frac{1}{\sqrt{\Delta p_{\mathrm{rare}}}}
\right).
\label{eq:SC_SQ_simplified}
\end{align}
When $N=\mathrm{poly}(1/\Delta)$, one has $\log N=\Theta(\log(1/\Delta))$, and therefore
\begin{equation}
\mathcal{S}_C(\Delta)
=
\mathcal{O}\left(
\frac{1}{\Delta}\log\frac{1}{\Delta}
+
\frac{1}{p_{\mathrm{rare}}}
\right).
\label{eq:SC_polyN}
\end{equation}
Thus, once $N$ is polynomial in $1/\Delta$, the asymptotic competition between classical and quantum algorithms is determined entirely by the scaling of $p_{\mathrm{rare}}$.

\subsection{Rank-frequency power-law model}

We consider a rank-ordered family of distributions over $N$ events,
\begin{equation}
P(x_k)=\frac{k^{-\gamma}}{Z_{N,\gamma}},
\label{eq:zipf_model}
\end{equation}
where $\gamma\ge 0$ is the tail exponent and $Z_{N,\gamma}:=\sum_{j=1}^{N} j^{-\gamma}$ is the generalized harmonic number (also the partition function).

The events are ordered so that $P(x_1)\ge P(x_2)\ge \cdots \ge P(x_N)$. Since the probabilities are monotone decreasing, the rare-event set is always a tail in rank space. For $\gamma>0$, define
\begin{equation}
k_\Delta := \min\{k\in\{1,\dots,N\}: P(x_k)\le \Delta\}.
\end{equation}
Following this, the rare-event set is $R_\Delta=\{x_k: k\ge k_\Delta\}$,
and the threshold rank satisfies
\begin{equation}
P(x_{k_\Delta})\le \Delta < P(x_{k_\Delta-1}).
\end{equation}

The total rare-event mass is
\begin{equation}
p_{\mathrm{rare}}
=
\sum_{k= k_\Delta}^{N} P(x_k)
=
\frac{1}{Z_{N,\gamma}}
\sum_{k= k_\Delta}^{N} k^{-\gamma}.
\label{eq:prare_sum}
\end{equation}
Because $k\mapsto k^{-\gamma}$ is monotone, integral comparison implies that the sum and integral have the same asymptotic scaling (a standard consequence of the integral test, also known as the Maclaurin–Cauchy test), especially when $N$ is large:
\begin{equation}
p_{\mathrm{rare}}
=
\Theta\left(
\frac{1}{Z_{N,\gamma}}
\int_{k_\Delta}^{N} x^{-\gamma} dx
\right).
\label{eq:prare_integral_scaling}
\end{equation}

To make the role of the state-space size explicit, we parameterize
\begin{equation}
N(\Delta)=\Theta(\Delta^{-q}),
\qquad q>0.
\label{eq:N_scaling}
\end{equation}
We also require that for sufficiently small $\Delta$,
\begin{equation}
P(x_N)\le \Delta < P(x_1).
\label{eq:nondegeneracy}
\end{equation}
The left inequality ensures that the rare set is nonempty, while the right inequality ensures that not every state is rare.

\subsection{Evaluation of the rare-event mass}

We now evaluate $p_{\mathrm{rare}}$ in the three tail regimes $\gamma>1$, $\gamma=1$, and $0<\gamma<1$. The case $\gamma=0$ is the uniform case, and one cannot reasonably distinguish which event is rare.

\textbf{Regime 1: $\gamma>1$.}
In this regime, the generalized harmonic number converges:
\begin{equation}
Z_{N,\gamma}
=
\zeta(\gamma)+\mathcal{O}(N^{1-\gamma})
=
\Theta(1),
\label{eq:H_gamma_gt_1}
\end{equation}
where $\zeta(\cdot)$ is the zeta function.
Therefore, the threshold rank obeys $k_\Delta
=
\Theta(\Delta^{-1/\gamma})$.

The nonemptiness condition $P(x_N)\le \Delta$ is equivalent to
\begin{equation}
\frac{N^{-\gamma}}{\zeta(\gamma)+o(1)} \lesssim \Delta,
\end{equation}
which, under $N=\Theta(\Delta^{-q})$, requires $q\ge 1/\gamma$ (up to constant factors at the boundary). In this regime, \cref{eq:prare_integral_scaling} yields
\begin{align}
p_{\mathrm{rare}}
&=
\Theta\left(
\int_{k_\Delta}^{N} x^{-\gamma} dx
\right)
=
\Theta\left(
\frac{k_\Delta^{ 1-\gamma}-N^{1-\gamma}}{\gamma-1}
\right).
\end{align}
Whenever $N\gtrsim k_\Delta$, the lower limit dominates, so
\begin{equation}
p_{\mathrm{rare}}
=
\Theta\left(k_\Delta^{ 1-\gamma}\right)
=
\Theta\left(\Delta^{(\gamma-1)/\gamma}\right),
\label{eq:prare_gt_1}
\end{equation}
implying that the rare-event mass vanishes polynomially as $\Delta\to 0$. Although the distribution has a long tail in cardinality, that tail carries too little total probability. The rare subspace therefore becomes increasingly difficult to amplify quantumly.

Substituting \cref{eq:prare_gt_1} into \cref{eq:SC_polyN}, we obtain
\begin{equation}
\mathcal{S}_C(\Delta)
=
\mathcal{O}\left(
\frac{1}{\Delta}\log\frac{1}{\Delta}
\right),
\label{eq:SC_gt_1}
\end{equation}
because
\begin{equation}
\frac{1}{p_{\mathrm{rare}}}
=
\Theta\left(\Delta^{-(\gamma-1)/\gamma}\right)
=
o(\Delta^{-1}).
\end{equation}
The quantum complexity is
\begin{align}
\mathcal{S}_Q(\Delta)
&=
\Theta\left(
\frac{1}{\sqrt{\Delta p_{\mathrm{rare}}}}
\right)
=
\Theta\left(
\Delta^{-1/2} 
\Delta^{-(\gamma-1)/(2\gamma)}
\right)
\nonumber\\
&=
\Theta\left(
\Delta^{-1+\frac{1}{2\gamma}}
\right).
\label{eq:SQ_gt_1}
\end{align}

Since $\gamma>1$, we have $-1+\frac{1}{2\gamma}<-\frac{1}{2}$, the exponent in \cref{eq:SQ_gt_1} is more negative than $-1/2$. In other words, the quantum algorithm is strictly worse than the ideal $\Theta(\Delta^{-1/2})$ scaling because it pays an additional amplification penalty coming from the shrinking mass of the rare subspace.
Therefore, the speedup is strictly sub-quadratic.

\textbf{Regime 2: $\gamma=1$.}
At the critical exponent $\gamma=1$, the normalization diverges logarithmically:
\begin{equation}
Z_{N,1}
=
\sum_{j=1}^N \frac{1}{j}
=
\log N + \gamma_{\mathrm E} + \mathcal{O}(1/N)
=
\Theta(\log N),
\label{eq:H_gamma_eq_1}
\end{equation}
where $\gamma_{\mathrm E}$ is the Euler's constant.
The threshold rank is therefore
\begin{equation}
k_\Delta
=
\Theta\left(\frac{1}{\Delta\log N}\right).
\label{eq:kdelta_eq_1}
\end{equation}
Using \cref{eq:prare_integral_scaling},
\begin{align}
p_{\mathrm{rare}}
&=
\Theta\left(
\frac{1}{\log N}
\int_{k_\Delta}^{N}\frac{dx}{x}
\right)
=
\Theta\left(
\frac{\log(N/k_\Delta)}{\log N}
\right).
\label{eq:prare_eq_1_general}
\end{align}

Now substitute $N=\Theta(\Delta^{-q})$. Then
\begin{equation}
\log N
=
q\log\frac{1}{\Delta}+\mathcal{O}(1),
\label{eq:logN_q}
\end{equation}
while \cref{eq:kdelta_eq_1} implies
\begin{equation}
\log k_\Delta
=
\log\frac{1}{\Delta}
-
\log\Big(q\log\frac{1}{\Delta}\Big)
+
\mathcal{O}(1).
\label{eq:logk_q}
\end{equation}
Hence
\begin{align}
\log\frac{N}{k_\Delta}
&=
\log N-\log k_\Delta
\nonumber\\
&=
(q-1)\log\frac{1}{\Delta}
+
\log\log\frac{1}{\Delta}
+
\mathcal{O}(1).
\label{eq:log_ratio_eq_1}
\end{align}
Substituting into \cref{eq:prare_eq_1_general} gives
\begin{equation}
p_{\mathrm{rare}}
=
\Theta\left(
\frac{
(q-1)\log(1/\Delta)+\log\log(1/\Delta)
}{
q\log(1/\Delta)
}
\right).
\label{eq:prare_eq_1_q}
\end{equation}

This formula reveals three distinct subregimes.

\begin{enumerate}[label=(\alph*)]
    \item $q>1$.  In this case,
\begin{equation}
p_{\mathrm{rare}}
\to
1-\frac{1}{q}
\in(0,1).
\label{eq:prare_eq_1_q_gt_1}
\end{equation}
Thus, the rare-event mass remains macroscopic. The classical and quantum complexities become
\begin{align}
\mathcal{S}_C(\Delta)
&=
\mathcal{O}\left(
\frac{1}{\Delta}\log\frac{1}{\Delta}
\right),\\
\mathcal{S}_Q(\Delta)
&=
\Theta\left(
\frac{1}{\sqrt{\Delta}}
\right).
\label{eq:complexity_eq_1_q_gt_1}
\end{align}
Therefore
\begin{equation}
\frac{\mathcal{S}_Q(\Delta)}{\sqrt{\mathcal{S}_C(\Delta)}}
=
\mathcal{O}\left(
\frac{1}{\sqrt{\log(1/\Delta)}}
\right)
\to 0.
\label{eq:strict_superquad_eq_1_q_gt_1}
\end{equation}
    \item $q=1$. This is the boundary case $N=\Theta(1/\Delta)$. \cref{eq:prare_eq_1_q} simplifies to
\begin{equation}
p_{\mathrm{rare}}
=
\Theta\left(
\frac{\log\log(1/\Delta)}{\log(1/\Delta)}
\right).
\label{eq:prare_eq_1_q_eq_1}
\end{equation}
Thus, the rare-event mass is \emph{not} constant, but it vanishes only logarithmically slowly. The quantum complexity becomes
\begin{equation}
\mathcal{S}_Q(\Delta)
=
\Theta\left(
\frac{1}{\sqrt{\Delta}}
\sqrt{
\frac{\log(1/\Delta)}{\log\log(1/\Delta)}
}
\right).
\label{eq:SQ_eq_1_q_eq_1}
\end{equation}
Meanwhile, the classical complexity remains
\begin{equation}
\mathcal{S}_C(\Delta)
=
\mathcal{O}\left(
\frac{1}{\Delta}\log\frac{1}{\Delta}
\right).
\label{eq:SC_eq_1_q_eq_1}
\end{equation}
This is a weaker boundary version of the phase transition.
    \item $q<1$. In this case, $P(x_N)\gg \Delta$, so the rare set is empty. This regime is outside the problem setting.

From the calculation, we see that in the regime of $\gamma=1$, the logarithmic divergence of the normalization exactly compensates the decay of the probabilities. This is the critical point at which the tail stops losing polynomial mass. If $N$ grows slightly faster than $1/\Delta$, the rare-event set already carries a constant fraction of the total probability, and the quantum algorithm reaches its ideal $\Theta(\Delta^{-1/2})$ scaling.
\end{enumerate}

\textbf{Regime 3: $0<\gamma<1$.}
We now turn to the most favorable heavy-tail regime for the quantum algorithm. 
For $0<\gamma<1$, the normalization diverges polynomially with $N$:
\begin{equation}
Z_{N,\gamma}
=
\sum_{j=1}^{N} j^{-\gamma}
=
\frac{N^{1-\gamma}}{1-\gamma}\big(1+o(1)\big).
\label{eq:H_gamma_lt_1_rewrite}
\end{equation}
Therefore, the rank-frequency law takes the asymptotic form
\begin{equation}
P(x_k)
=
\frac{k^{-\gamma}}{Z_{N,\gamma}}
=
(1-\gamma) N^{-(1-\gamma)} k^{-\gamma}\big(1+o(1)\big).
\label{eq:P_gamma_lt_1_rewrite}
\end{equation}

Since $P(x_k)$ is monotone decreasing in $k$, define the threshold rank by
\begin{equation}
k_\Delta
:=
\min\{k\in\{1,\dots,N\}: P(x_k)\le \Delta\}.
\label{eq:kDelta_definition_regime3}
\end{equation}
Whenever $k_\Delta\le N$, the rare-event set is precisely
\begin{align}
R_\Delta=\{x_k: k\ge k_\Delta\}.
\end{align}

Using \cref{eq:P_gamma_lt_1_rewrite}, the condition $P(x_k)\le \Delta$ is equivalent, to leading order, to
\begin{align}
(1-\gamma) N^{-(1-\gamma)} k^{-\gamma}\lesssim \Delta,
\end{align}
which gives
\begin{equation}
k_\Delta
=
\left(
\frac{1-\gamma}{\Delta N^{1-\gamma}}
\right)^{1/\gamma}
\big(1+o(1)\big).
\label{eq:kDelta_regime3}
\end{equation}
Dividing by $N$, we obtain the especially useful ratio
\begin{equation}
\frac{k_\Delta}{N}
=
\left(
\frac{1-\gamma}{\Delta N}
\right)^{1/\gamma}
\big(1+o(1)\big).
\label{eq:kDelta_over_N_regime3}
\end{equation}

This expression makes the nondegeneracy conditions transparent. To have a genuine rare-event problem, we must require both:
\begin{enumerate}
    \item the rare set is non-empty, and
    \item the rare set does not coincide with the entire support.
\end{enumerate}

The first requirement is that at least one state lies below the threshold. Because $P(x_k)$ decreases with $k$, this is equivalent to demanding that the least probable state be rare:
\begin{equation}
P(x_N)\le \Delta.
\label{eq:nonempty_rare_condition}
\end{equation}
Using \cref{eq:P_gamma_lt_1_rewrite} at $k=N$, we find
\begin{equation}
P(x_N)
=
\frac{N^{-\gamma}}{Z_{N,\gamma}}
=
\frac{1-\gamma}{N}\big(1+o(1)\big).
\label{eq:PN_regime3}
\end{equation}
Hence, the rare set is non-empty only if $(1-\gamma)/N\lesssim \Delta$, equivalently, $N\gtrsim (1-\gamma)/\Delta$.
If we parameterize the support size as
\begin{equation}
N(\Delta)=\Theta(\Delta^{-q}),
\label{eq:N_scaling_regime3}
\end{equation}
then it implies, at the exponent level, $q\geq 1$.
This lower bound is therefore not an extra assumption; it is exactly the condition ensuring that the rare set is not empty. More precisely:
\begin{itemize}
    \item if $q<1$, then $P(x_N)\gg \Delta$, so no state is rare asymptotically;
    \item if $q=1$, writing $N\sim c/\Delta$, non-emptiness requires $c\ge 1-\gamma$;
    \item if $q>1$, then $P(x_N)\ll \Delta$, so the rare set is automatically non-empty.
\end{itemize}

The second requirement is that not every state be rare. Since $P(x_1)$ is the largest probability, this means $P(x_1)>\Delta$.
Using \cref{eq:P_gamma_lt_1_rewrite} at $k=1$, we obtain
\begin{equation}
P(x_1)
=
\frac{1}{Z_{N,\gamma}}
=
\frac{1-\gamma}{N^{1-\gamma}}\big(1+o(1)\big).
\label{eq:P1_regime3}
\end{equation}
Thus, to keep the problem nontrivial, we need
\begin{equation}
\frac{1-\gamma}{N^{1-\gamma}}\gtrsim \Delta.
\label{eq:not_all_rare_condition2}
\end{equation}
Under the scaling ansatz $N(\Delta)=\Theta(\Delta^{-q})$, this becomes
\begin{equation}
q<\frac{1}{1-\gamma}
\label{eq:q_upper_bound_regime3}
\end{equation}
at the exponent level. Indeed, if $q>1/(1-\gamma)$, then even the most probable state satisfies $P(x_1)\ll \Delta$, so all states are rare asymptotically. The boundary case $q=1/(1-\gamma)$ depends on prefactors and is excluded here for simplicity.

Combining these two requirements, the genuinely nondegenerate regime is
\begin{equation}
1\le q<\frac{1}{1-\gamma},
\label{eq:nondegenerate_q_regime3}
\end{equation}
with the lower boundary $q=1$ requiring the prefactor condition $N\sim c/\Delta$ and $c>1-\gamma$ if one wants a strictly positive limiting rare mass.

We now evaluate the aggregate rare-event mass. By integral comparison,
\begin{align}
p_{\mathrm{rare}}
&=
\frac{1}{Z_{N,\gamma}}
\sum_{k= k_\Delta}^{N} k^{-\gamma}
\nonumber\\
&=
\frac{1}{Z_{N,\gamma}}
\left(
\frac{N^{1-\gamma}-k_\Delta^{1-\gamma}}{1-\gamma}
\right)
+o(1).
\label{eq:prare_regime3_step1}
\end{align}
Using \cref{eq:H_gamma_lt_1_rewrite}, this simplifies to
\begin{equation}
p_{\mathrm{rare}}
=
1-\left(\frac{k_\Delta}{N}\right)^{1-\gamma}+o(1).
\label{eq:prare_regime3_step2}
\end{equation}
Finally, substituting \cref{eq:kDelta_over_N_regime3} gives
\begin{equation}
p_{\mathrm{rare}}
=
1-
\left(
\frac{1-\gamma}{\Delta N}
\right)^{\frac{1-\gamma}{\gamma}}
+o(1).
\label{eq:prare_regime3_final}
\end{equation}

There are two relevant nondegenerate subcases:
\begin{enumerate}[label=(\alph*)]
    \item $q=1.$
    Write $N\sim c/\Delta$ with $c>1-\gamma$. Then
\begin{equation}
p_{\mathrm{rare}}
\to
1-
\left(
\frac{1-\gamma}{c}
\right)^{\frac{1-\gamma}{\gamma}}
\in(0,1).
\label{eq:prare_regime3_q1}
\end{equation}
Thus, even at the minimal support growth required for non-emptiness, the rare tail already carries a strictly positive constant fraction of the total probability mass, provided $c>1-\gamma$.
\item $1<q<1/(1-\gamma).$
Here $\Delta N\to\infty$, so \cref{eq:prare_regime3_final} gives
\begin{equation}
p_{\mathrm{rare}}\to 1.
\label{eq:prare_regime3_qgt1}
\end{equation}
In other words, as soon as the support grows faster than $1/\Delta$ but still remains below the full-degeneracy threshold, almost all of the probability mass lies in the rare tail.
\end{enumerate}

Therefore, throughout the entire nondegenerate regime in \cref{eq:nondegenerate_q_regime3},
\begin{equation}
p_{\mathrm{rare}}=\Omega(1).
\label{eq:prare_Omega1_regime3}
\end{equation}

This is the decisive structural property of the $0<\gamma<1$ regime. Although each individual rare event has a probability at most $\Delta$, the set of all such events retains a macroscopic fraction of the total mass. The tail is therefore heavy not only in cardinality, but also in aggregate probability. As a result, the amplitude-amplification stage of the quantum algorithm does not suffer any asymptotic penalty from a vanishing target subspace.

Substituting $p_{\mathrm{rare}}=\Omega(1)$ into the general complexity bounds yields
\begin{align}
\mathcal{S}_C
&=
\mathcal{O}\left(
\frac{1}{\Delta}\log\frac{1}{\Delta}
\right),\\
\mathcal{S}_Q
&=
\mathcal{O}\left(
\frac{1}{\sqrt{\Delta}}
\right).
\end{align}
Hence the quantum algorithm retains its ideal $\Delta^{-1/2}$ scaling, while the classical method still pays the logarithmic identification cost.

\subsection{Summary of quantum advantages and interpretations}

The preceding case analysis shows that the critical quantity is not merely the number of rare states, but the total probability mass carried by those states. The quantum algorithm always pays the thresholding cost $\Theta(\Delta^{-1/2})$, but it pays an additional amplification factor $p_{\mathrm{rare}}^{-1/2}$. Thus:
\begin{itemize}
    \item if $p_{\mathrm{rare}}$ vanishes polynomially, as it does for $\gamma>1$, the speedup is degraded to sub-quadratic;
    \item if $p_{\mathrm{rare}}=\Omega(1)$, as it does generically for $0<\gamma<1$ and for $\gamma=1$ with $q>1$, the quantum complexity saturates at $\Theta(\Delta^{-1/2})$;
    \item at the boundary $(\gamma,q)=(1,1)$, the rare mass shrinks only logarithmically, quantum algorithm does not achieve the ideal $\Theta(\Delta^{-1/2})$ form.
\end{itemize}

This identifies a sharp phase transition at $\gamma=1$ \emph{within the power-law family of \cref{eq:zipf_model}}, and we list all the results in \cref{tab:complexity_scaling_corrected}.
Above the critical exponent, the rare tail is too light and $p_{\mathrm{rare}}$ collapses. At and below the critical exponent, the tail remains macroscopically populated, and the amplitude-amplification penalty disappears.

The argument proves a phase transition \emph{within the rank-frequency power-law model}. We do not claim that $\gamma\le 1$ is the only possible way for an arbitrary distribution family to satisfy $p_{\mathrm{rare}}=\Omega(1)$ together with $N=\mathrm{poly}(1/\Delta)$. For the Zipfian power-law families analyzed here, the threshold $\gamma=1$ is where the asymptotic behavior changes.

\begin{table}[t]
\centering
\small
\renewcommand{\arraystretch}{1.22}
\setlength{\tabcolsep}{4.5pt}
\begin{threeparttable}
\caption{Asymptotic regimes for rank-frequency power laws
$P(x_k)\propto k^{-\gamma}$ with $N(\Delta)=\Theta(\Delta^{-q})$.}
\label{tab:complexity_scaling_corrected}
\begin{tabular}{@{}ccccc@{}}
\toprule
\textbf{Tail regime}
&
\textbf{Condition on $q$}
&
\textbf{Rare mass $p_{\mathrm{rare}}$}
&
\textbf{Quantum cost $\mathcal{S}_Q(\Delta)$}
&
\textbf{Advantage}
\\
\midrule
$\gamma>1$
&
$q\ge 1/\gamma$
&
$\Theta\left(\Delta^{(\gamma-1)/\gamma}\right)$
&
$\Theta\left(\Delta^{-1+\frac{1}{2\gamma}}\right)$
&
Subquadratic
\\
$\gamma=1$
&
$q>1$
&
$1-\frac{1}{q}+o(1)$
&
$\Theta \left(\Delta^{-1/2}\right)$
&
Ideal quadratic
\\
$\gamma=1$
&
$q=1$
&
$\Theta \left(
\frac{\log\log(1/\Delta)}{\log(1/\Delta)}
\right)$
&
$\Theta \left(
\Delta^{-1/2}
\sqrt{
\frac{\log(1/\Delta)}{\log\log(1/\Delta)}
}
\right)$
&
Nearly quadratic
\\
$0<\gamma<1$
&
$1\le q<1/(1-\gamma)$
&
$\Omega(1)$
&
$\Theta \left(\Delta^{-1/2}\right)$
&
Ideal quadratic
\\
$\gamma=0$
&
--
&
Degenerate
&
Degenerate
&
Degenerate
\\
\bottomrule
\end{tabular}
\begin{tablenotes}[flushleft]
\footnotesize
\item Here $\mathcal{S}_Q(\Delta)$ suppresses logarithmic factors in
$\log(1/\epsilon)$. In the nondegenerate regimes, the classical
cost satisfies
$\mathcal{S}_C(\Delta)=O \left(\Delta^{-1}\log(1/\Delta)\right)$
and $\mathcal{S}_C(\Delta)=\Omega(\Delta^{-1})$.
The transition occurs at $\gamma=1$.
\end{tablenotes}
\end{threeparttable}
\end{table}

\section{Application to stochastic process}

\subsection{Stochastic processes and the Asymptotic Equipartition Property}
\label{app:stochastic&AEP}

A discrete stochastic process generates a random output $X_t$ at each time step $t$, which takes value $x_t$ from an alphabet $\mathcal{X}$ of finite alphabet size. A consecutive $L$ output sequence $x_{t:t+L}:= x_t,x_{t+1},\cdots x_{t+L-1}$ is governed by the joint probability distribution $\Pr(x_{t:t+L})$.

Stochastic processes can be effectively simulated using recurrent models, which comprise a memory system and an output mechanism~\cite{rabiner2003introduction}. At time step $t$, given a memory state $S_i$, the model generates an output $x_t$ and updates its memory to state $S_j$ according to the joint probability $P( S_j, x_t | S_i)$. Following this update, the model recursively applies these transition rules to generate the subsequent output sequence.

\begin{definition}[Joint entropy]
    Given a collection of random variables $X_1,\dots, X_L$, with support on $\Omega_1,\dots,\Omega_L$ respectively, we define the joint entropy of the collection of discrete random variables as
    \begin{align}
        H(X_1,\dots,X_L)=-\sum_{x_1\in \mathcal{X}}\cdots\sum_{x_L\in \mathcal{X}} P(x_1,\dots,x_L)\log P(x_1,\dots, x_L).
    \end{align}
\end{definition}

For discrete-time stochastic processes, the entropy rate is defined asymptotically.
However, note that the entropy rate may not exist for an arbitrary stochastic process.

\begin{definition}[Entropy rate of stochastic process]
    For a discrete-valued, discrete-time stochastic process $\mathfrak{X}$, the entropy rate is defined as
    \begin{align}
        H(\mathfrak{X})=\lim_{L\rightarrow \infty} \frac{1}{L}H(X_{0:L}).
    \end{align}
\end{definition}

Here, we focus on the so-called stationary distribution, where any consecutive $L$ outputs are statistically invariant with respect to time,
\begin{equation}
    P(x_{0:L}) = P(x_{t:t+L}) \quad \forall t, L.
\end{equation}

For a stationary process, it can be shown that the entropy rate exists and can be quantified via the conditional entropy.
\begin{lemma}
    For a stationary stochastic process, the entropy rate exists and is equal to
    \begin{align}
        H(\mathfrak{X})=\lim_{L\rightarrow \infty} H(X_{L-1}|X_{L-2},\dots, X_0).
    \end{align}
\end{lemma}

The following theorem states that for a well-conditioned stochastic process, almost all sequences we expect to see a probability $p(x_{0:L})\approx 2^{-LH(\mathfrak{X})}$.
A stationary stochastic process is ergodic if the time average of the output sequence converges to the ensemble average (or space average), namely,
\begin{equation}
    \lim_{L\to \infty} \frac{1}{L}\sum_{t=1}^L X_t = \mathbb{E}[X].
\end{equation}

\begin{lemma}[Asymptotic Equipartition Property (AEP) \cite{algoet1988Sandwicha}]
    Assume $\mathfrak{X}$ is a stationary and ergodic process, the probability of a sequence $x_{0:L}$ converges to the entropy rate of the stochastic process, that is,
    \begin{equation}
        \lim_{L\to \infty}-\frac{1}{L} \log P(x_{0:L}) = H(\mathfrak{X}).
    \end{equation}
\end{lemma}

This can be understood as an information-theoretic analog of the law of large numbers.
We define the typical sequence to have a tighter bound for these sequences.

\begin{definition}[$\epsilon_t$-typical sequence]
    For $\epsilon_t>0$, we say a $L$-length sequence $x_{0:L}$ is $\epsilon_t$-typical if
    \begin{align}
        \left|-\frac{1}{L}\log P(x_{0:L})-H(\mathfrak{X})\right|\leq \epsilon_t.
    \end{align}
\end{definition}
We define typical set $\mathcal{T}(\epsilon_t,L)$ as the set of all $\epsilon_t$-typical $L$-length sequences, that is,
\begin{equation}
    \mathcal{T}(\epsilon_t,L) \coloneqq \left\{ x_{0:L} : \left| - \frac{1}{L}\log P(x_{0:L}) - H(\mathfrak{X}) \right| \leq \epsilon_t\right\}.
\end{equation}

The AEP theorem implies the following theorem about the typical set:
\begin{lemma}[Properties of typical sequence]\label{thm.typical}
    Let $\epsilon_t>0$ be fixed.
    For any $\delta_t>0$, for sufficiently large $L$, the probability that a sequence is $\epsilon_t$-typical is at least $1-\delta$, namely,
    \begin{align}
        P(x_{0:L}\in \mathcal{T}(\epsilon_t,L) ) \geq 1-\delta_t,
    \end{align}
    and the size of the typical set satisfies
    \begin{align}
        (1-\delta_t)2^{L(H(\mathfrak{X}) - \epsilon_t)} \leq \left| \mathcal{T}(\epsilon_t,L) \right| \leq 2^{L(H(\mathfrak{X}) + \epsilon_t)}.
    \end{align}
\end{lemma}

An important observation is that the size of the typical set is exponentially smaller than the whole possible set, unless the stochastic process is nearly uniformly distributed. This can be seen from evaluating the ratio between the size of the typical set and the size of all events.
\begin{align}
    \frac{| \mathcal{T}(\epsilon_t,L)|}{|\mathcal{X}|^L}\leq \frac{2^{L(H(\mathfrak{X}) + \epsilon_t)}}{|\mathcal{X}|^L}=\frac{2^{L(H(\mathfrak{X}) + \epsilon_t)}}{2^{L \log |\mathcal{X}|}}=2^{L(H(\mathfrak{X}) + \epsilon_t-\log |\mathcal{X}|)}.
\end{align}

A $\varepsilon$-machine is a special type of Hidden Markov Model that consists of several hidden states, which serve as the memory states $S_i$ that carry information from the past.
The past is defined as a semi-infinite sequence $\overleftarrow{x}:= x_{-\infty:0}$. An epsilon machine maps the past $\overleftarrow{x}$ to one of the hidden states $S_i$. And at each time step, the machine generates an output $x$ and updates its memory state from $S_i$ to $S_j$ with probability $\Pr(x, S_j|S_i)$. Such a $\varepsilon$-machine is also called a unifilar hidden Markov model whose output $x_t$ and the hidden state $S_t$ determine the next hidden state $S_{t+1}$ uniquely.

\subsection{Thermodynamic mapping for rare events}
\label{app:ratio}

By definition, rare events are events that are less likely to happen than typical events.
The AEP theorem, per \cref{thm.typical}, claims that all of the typical sequences happen with probability roughly $2^{-LH(\mathfrak{X})}$.
\begin{definition}[Rare event for stochastic process]\label{def.rare.stoq}
    For a stochastic process $\mathfrak{X}$ with entropy rate $H(\mathfrak{X})$, we define the set of rare events as
    \begin{align}
        \mathcal{R}_K^L = \{x_{0:L} | P(x_{0:L}) \leq  2^{-KL H(\mathfrak{X})}\}\coloneqq \{x_{0:L} | P(x_{0:L}) \leq  \Delta\},
    \end{align}
    where $K>1$ is some large constant, and $\Delta = 2^{-L KH(\mathfrak{X})}$.
\end{definition}

To verify whether the event is rare or not, there is a small range of uncertainty.
\begin{definition}[Unsure event for stochastic process]
    For a stochastic process $\mathfrak{X}$ with entropy rate $H(\mathfrak{X})$, we define the set of unsure events as
    \begin{align}
        \mathcal{S}_\xi = \{x_{0:L}| \Delta \leq \Pr(x_{0:L}) \leq \Delta +2\xi\},
    \end{align}
    where $\Delta = 2^{-L KH(\mathfrak{X})}$, and $0<\xi<\Delta$.
\end{definition}

To bound the ratio of rare events is to notice that the entropy rate is similar to the ``energy density", inspired by the concept of Boltzmann weight in statistical physics. By changing the ``temperature", it is possible to construct a mapping, such that the atypical set of the original process is then transformed into the typical set of the new process \cite{Aghamohammadi2017Minimuma}.
Analogously to the energy, we write $u = KH(\mathfrak{X})$.
We define a set
\begin{equation}
    \Lambda_{u,L}^\mathfrak{X} = \left\{x_{0:L}\middle| - \frac{\log_2\Pr(x_{0:L})}{L} = u\right\}.
\end{equation}
\begin{lemma}[$\beta$-mapping \cite{Aghamohammadi2017Minimuma}]\label{lemma:beta.mapping}
    For the stationary stochastic process $\mathfrak{X}$, there exists a new stochastic process $\mathfrak{X}_\beta$ such that
    \begin{equation}
        \lim_{L\to\infty}\Lambda_{u,L}^\mathfrak{X} = \lim_{L\to\infty}\Lambda_{u_\beta,L}^{\mathfrak{X}_\beta},
    \end{equation}
    where $u_\beta = \beta u +\log_2 \lambda$, and $\beta, \lambda>0$.
    For any $u$, there exists a $\beta$ such that $u_\beta = H(\mathfrak{X}_\beta)$.
\end{lemma}

With this lemma in hand, we can bound the size of the rare event set $\mathcal{S}_\xi$.

\begin{theorem}[Bounds for the rare event]\label{thm.ratio}
    For sufficiently large $L$, the probability $p_{rare}$ rare event set can be lower bounded by
    \begin{equation}
        p_{\mathrm{rare}} =\Omega( \Delta^{1 - \mu})
    \end{equation}
    where $\mu = \frac{H(\mathfrak{X}_\beta)}{KH(\mathfrak{X})} <1$.
\end{theorem}

\begin{proof}
    According to \cref{lemma:beta.mapping}, given a $u$, we can find a $\beta$ such that
    \begin{equation}
        \lim_{L\to\infty}\Lambda_{u,L}^\mathfrak{X} = \lim_{L\to\infty}\Lambda_{u_\beta,L}^{\mathfrak{X}_\beta}, u_\beta = H(\mathfrak{X}_\beta).
    \end{equation}
    We can define an alternative set 
    \begin{equation}
        \Lambda_{u,L}^{\mathfrak{X},\epsilon} = \left\{x_{0:L}\middle| |- \frac{\log_2\Pr(x_{0:L})}{L} - u| \leq \epsilon \right\}.
    \end{equation}

    For any tolerant error $\epsilon_2$, we can find a tolerant error $\epsilon_1$ and sequence length $L$, such that the typical set $\mathcal{T}_{\beta}(\epsilon_1,L)$ of process $\mathfrak{X}$   satisfies
    \begin{equation}
        \mathcal{T}_{\beta}(\epsilon_1,L) \subseteq  \Lambda_{u,L}^{\mathfrak{X},\epsilon_2}
    \end{equation}

    Let $u= KH(\mathfrak{X}) + \eta$, where $\eta$ is a very small shift. We choose $\epsilon_2 <\eta$. Then we can obtain
    \begin{equation}
        \Lambda_{u,L}^{\mathfrak{X},\epsilon_2} \subseteq \mathcal{R}_K^L
    \end{equation}
    Then $\mathcal{T}_{\beta}(\epsilon_1,L) \subseteq \mathcal{R}_K^L$. 

    By \cref{thm.typical}, we can have the following lower bound
    \begin{align}
        P(\Lambda_{u,L}^{\mathfrak{X},\epsilon_2} ) &\geq 2^{-L(KH(\mathfrak{X}) + \eta +\epsilon_2)} |\mathcal{T}_{\beta}(\epsilon_1,L)| \\
        & \geq 2^{-L(KH(\mathfrak{X}) + \eta +\epsilon_2)}(1-\delta)2^{L (H(\mathfrak{X}_\beta)-\epsilon_1)}\\
        & = \Omega ( \Delta^{1 - \mu})
    \end{align}
   where $\mu = \frac{H(\mathfrak{X}_\beta)}{KH(\mathfrak{X})} <1$.

   Since $\Lambda_{u,L}^{\mathfrak{X},\epsilon_2} \subseteq \mathcal{R}_K^L$, we can have the lower bound of the probability of rare event set,
   \begin{equation}
       p_{\mathrm{rare}} \geq P(\Lambda_{u,L}^{\mathfrak{X},\epsilon_2} ) =  \Omega ( \Delta^{1 - \mu})
   \end{equation}
\end{proof}

\subsection{Quantum rare event sampling for stochastic processes}
\label{app.stochastic.alg}

In the following, we detail the applications of our algorithm to stochastic processes, which have additional properties provided by the AEP theorem and $\beta$-mapping.
This corresponds to a special case of \cref{thm.alg.rare.pure} when $Q=P$ and $p_1-p_2=\mathcal{O}(1)$.

Given sampling access to a stochastic process, our goal is to sample rare events whose probability is significantly smaller than that of the typical events.
This goal is similar to the low-energy state preparation in quantum physics. One may regard a given distribution $P(x_i)$ as the energy of a physical system's state $x_i$.
Consequently, the physical system has the Hamiltonian $H=\sum_{i} \sqrt{P(x_i)}\ketbra{x_i}{x_i}$. Then, finding the low-energy states of the Hamiltonian, whose energy is less than $\sqrt{\Delta}$, is equivalent to finding the rare events of the stochastic process.

\begin{corollary}[Quantum rare event sampling for stochastic process]\label{thm.alg.rare}
    Consider a stochastic process $\mathfrak{X}$ satisfying the asymptotic equipartition property and the $\beta$-mapping.
    Let the rare event set $\mathcal{R}^L_K$ be defined as \cref{def.rare.stoq} and $p_{\mathrm{rare}}$ be the sum of probabilities of all rare events.
    Let the ambiguity be $\xi$, and the ratio of the probability of unsure events to rare events be $\zeta$.
    Given access to a sampler $U_L$ for this stochastic process, one can construct a quantum sampler for a distribution that is $\mathcal{O}(\epsilon+\zeta)$-close to the distribution
    \begin{align}
        P_{\mathcal{R}}= \begin{cases}
            P(x_{i}) /{p_{\mathrm{rare}}}, & x_{i}\in \mathcal{R}^L_K,\\
            0, & x_{i}\notin \mathcal{R}^L_K,
            \end{cases}
    \end{align}
    by using $\mathcal{O}( \sqrt{\frac{\Delta^{\mu}}{\xi^2}}\log  \frac{1}{\epsilon})$ times of $U_L$, where $\mu=\frac{H(\mathfrak{X}_\beta)}{KH(\mathfrak{X})}<1$.
\end{corollary}

\begin{proof}
    The algorithm follows mostly the same as \cref{thm.alg.rare.pure}.
    To make sure the ambiguity range is $\xi$, for quantum we need to set $\xi'$ such that
    \begin{align}
        (\sqrt{\Delta}+2\xi')^2 = \Delta + 2\xi.
    \end{align}
    It suffices to take $\xi'=\frac{-\sqrt{\Delta}+\sqrt{\Delta+2 \xi}}{2}$.
    Since $\xi$ is small compared to $\Delta$, we further have $\xi'=\Theta(\xi/\sqrt{\Delta})$.
    \cref{approx.rectangle} guarantees that the rectangular function can be approximated to error $\epsilon$ by an even polynomial $P$ with degree $\mathcal{O}(\frac{1}{\xi'}\log(1/\epsilon))$.
    \cref{theorem.qsvt} claims that such an even polynomial can be applied to the diagonal matrix $A$ by using $\mathcal{O}(\frac{1}{\xi'}\log(1/\epsilon))$ queries to controlled-$U_L$ and controlled-$U_L^\dagger$.

    We can set $t= \sqrt{\Delta}+\xi', \Gamma =\xi'$ in \cref{approx.rectangle}, then construct a construct a $(1, n+3, \epsilon)$-encoding $U_{\Pi_\mathcal{R}}$ of a reflection operator by using $\mathcal{O}(\frac{1}{\xi'}\log(\frac{1}{\epsilon}))$ times of controlled-$U$ and controlled-$U^\dagger$. The unsure set in \cref{approx.rectangle} is
    \begin{equation}
        \mathcal{S}_{\xi}=\{x_{0:L}|\Delta \leq \Pr(x_{0:L}) \leq \Delta + 2\xi \}.
    \end{equation}
    Let $\Pi_\mathcal{R}$ denote the projector into all rare events and $\Pi_{\mathrm{un}}$ denote the projector into the unknown area, i.e.,
    \begin{equation}
        \Pi_\mathcal{R} = \sum_{x_{0:L}\in \mathcal{R}_K^L}\ket{x_{0:L}}\bra{x_{0:L}},~ \Pi_{\mathrm{un}} = \sum_{x_{0:L}\in \mathcal{S}_\xi}\ket{x_{0:L}}\bra{x_{0:L}}.
    \end{equation}

    Then, we apply the block encoding $U_{\Pi_\mathcal{R}}$ to the initial state
    \begin{align}
        U_{\Pi_\mathcal{R}} \ket{0}\sum_{x_{0:L}} \sqrt{P(x_{0:L})}\ket{x_{0:L}} & \approx \ket{0}\Pi_\mathcal{R} \sum_{x_{0:L}} \sqrt{P(x_{0:L})}+\ket{1}\ket{\widetilde{\perp}}\nonumber \\ &=\ket{0}\sum_{x_{0:L}\in \mathcal{R}^{L}_{K}}\sqrt{P(x_{0:L})} \ket{x_{0:L}} +\ket{1}\ket{\widetilde{\perp}},
    \end{align}
    where $\ket{\widetilde{\perp}}$ is an unnormalized state.
    The probability of measuring the ancilla qubits  as $\ket{0}$ is
    \begin{equation}
    \|\tilde{\Pi}_\mathcal{R}\ket{\phi} \|^2 \geq \|\Pi_\mathcal{R}\ket{\phi} \|^2 (1-\epsilon)^2 = p_{\mathrm{rare}}(1-\epsilon)^2.
    \end{equation}
    By further using the amplitude amplification $\mathcal{O}(\frac{1}{\sqrt{p_{\mathrm{rare}}}})$ times, one can prepare the target state
    \begin{align}
        \frac{1}{\sqrt{p_{\mathrm{rare}}}}\sum_{x_{0:L}\in \mathcal{R}^{L}_{K}}\sqrt{P(x_{0:L})}\ket{x_{0:L}}.
    \end{align}

    Combining the QSVT procedure requiring $\mathcal{O}(1/\xi' \log \frac{1}{\epsilon})$, the overall queries to the controlled-$U$ and controlled-$U^\dagger$ is $\mathcal{O}(\frac{\sqrt{\Delta}}{\xi\sqrt{ p_{\mathrm{rare}}}}\log\frac{1}{\epsilon})$. Since $p_{\mathrm{rare}}= \Omega(\Delta^{1-\mu})$ where $\mu <1$, the final complexity is 
    $$
    \mathcal{O}\left(\frac{\Delta^{\mu/2}}{\xi}\log\frac{1}{\epsilon}\right)
    $$

    Following exactly the same error analysis as in \cref{thm.alg.rare.pure}, one can show the total variation distance is $\mathcal{O}(\epsilon+\zeta)$.
\end{proof}

In most cases $\xi=\mathcal{O}(\Delta)$, and we obtain the complexity as $\mathcal{O}(\sqrt{\Delta^{\mu -2}}\log\frac{1}{\epsilon})$.
Under the same conditions, following \cref{thm.alg.classical}, the classical complexity would be $\mathcal{O}(\frac{1}{\Delta}\log \frac{1}{\epsilon})$, which yields at least a polynomial speedup.

\subsection{Quantum stochastic modeling}

Given a stochastic process governed by distribution $\Pr(x_{0:L})$, we can always find a unifilar hidden Markov model, called $\varepsilon$-machine. The $\varepsilon$-machine
consists of a set of hidden memory states $S_i$with transition probability $\Pr(S_j,x|S_i)$. At each time step, the memory state updates to $S_j$ from $S_i$ and outputs $x$ with probability $\Pr(S_j,x|S_i)$. The $\varepsilon$-machine is unifilar, meaning that the current memory state $S_i$ and output $x$ determine the next memory state $S_j$.

From the $\varepsilon$-machine, we can construct a quantum model. The first step is to construct the MPS state
\begin{equation}
    A_{x,ij} = \sqrt{\Pr(S_j,x|S_i)}.
\end{equation}
Then, we can evaluate its left canonical form $\tilde{A}_x$ such that
\begin{equation}
    \sum_x \tilde{A}_x^\dagger \tilde{A}_x = I.
\end{equation}
$\tilde{A}_x$ forms a set of Kraus operators that can be used to construct the unitary operator $U$ such that
\begin{equation}
    U\ket{i}\ket{0} = \sum_x \tilde{A}_x\ket{i}\ket{x}.
\end{equation}

The quantum model will also map the classical state $S_i$ to a quantum state $\ket{\sigma_i}$, which are usually non-orthogonal to each other.
The memory states $\ket{\sigma_i}$ can also be computed from MPS states $A_{x,ij}$.
It is proven that for certain stochastic processes, the memory dimension $d_q $ of the quantum model, which is defined to be the dimension of the Hilbert space that the quantum states lie in, can be significantly less than the logarithm of the number of classical states $S_i$~\cite{yang2018matrix}.

Given an initial state $\ket{\sigma_i}$, we can couple the memory state with $L$ output systems initialized in state $\ket{0}^{\otimes L}$, which results in
\begin{equation}
    U_L\ket{\sigma_i}\ket{0}^{\otimes L} = \sum_{x_{0:L}} \sqrt{\Pr(x_{0:L}|S_i)}\ket{\sigma_j}\ket{x_{0:L}},
\end{equation}
where $\ket{\sigma_j}$ denotes the subsequent memory state determined by the initial state $\ket{\sigma_i}$ and the output sequence $x_{0:L}$.

Meanwhile, the decoupling procedure decouples the memory state from the output systems by applying the reverse unitary $U_{r}$
Assume the stochastic process has Markov order $\chi$, which means $\ket{\sigma_i}$ is determined by the last $\chi_i$ bits. Then we can construct the reverse unitary $U_{r}$ which is a $\chi$-qubit control unitary such that
\begin{equation}
    U_{r} = \sum_{x_{0:\chi}} \ket{x_{0:\chi}}\bra{x_{0:\chi}}\otimes U_{\sigma_i}^\dagger
\end{equation}
where $U_{\sigma_i} \ket{0} = \ket{\sigma_i}$. Applying $U_{r}$ will map the memory state back to $\ket{0}$. Therefore, combining $U_L$ and $U_{r}$ will give the desired quantum state
\begin{equation}
    U_{r}U_L\ket{\sigma_i}\ket{0}^{\otimes L} = \ket{0} \left(\sum_{x_{0:L}} \sqrt{\Pr(x_{0:L}|S_i)}\ket{x_{0:L}}\right).
\end{equation}

\section{Numerical results}
\label{app:numerics}
\subsection{Setup}
We simulate the results of our algorithm using \texttt{numpy} by directly obtaining the matrix representation of the block-encoding unitaries and operating on them.

For simulations on small-scale systems, we first construct a $(1,(L+\chi)\log|\mathcal{X}|+2, 0)$-block encoding where diagonal encodes the amplitudes of the output quantum state of the recurrent quantum circuit via the protocol as listed in the main paper. Note that the additional qubits stem in the block encoding when compared to \cref{thm.amplitudes.encoding} from the fact that the memory qubit is reverted to the zero state, and serves as an ancilla in the block encoding. The amplitudes are then accordingly transformed via our algorithm to produce the required block encoding given an extra ancilla qubit.

For larger systems, however, constructing the full $(1,(L+\chi)\log|\mathcal{X}|+2, 0)$-block encoding requires an additional polynomial cost for classical statevector simulators due to the much larger statevector from the excess ancilla qubits. We instead classically construct a $(1, \chi\log|\mathcal{X}|, 0)$-block encoding unitary as follows:
\begin{equation}
    U(D) = \begin{pmatrix} D & \sqrt{I-D^2}\\ \sqrt{I-D^2} & -D\end{pmatrix}
\end{equation}
where $D$ is a diagonal square matrix that encodes the amplitudes of the output state of the recurrent quantum circuit on its diagonal. While an efficient implementation of such a block encoding on a quantum circuit is an open question, for simulation purposes, it serves the same effect as the diagonal block encoding we construct in our paper, given that the $\mathrm{PREP}^\dagger\cdot\mathrm{SEL}\cdot\mathrm{PREP}$ methodology used to obtain the diagonal block encoding in \citep{rattew2023nonlinear} produces a Grover-like reflection operator. This smaller-scale construction can be used to simulate much larger systems with much less time and memory consumption than a full implementation of the block encoding.

To obtain the phase angles in the QSVT algorithm that approximates the thresholding function, we use the optimization method introduced in \citep{dong2021efficient} and implemented in \texttt{pyqsp}~\citep{chao2020finding,martyn2021grand} package. Instead of the constructive methods used to construct angles typically used in proofs for QSVT~\cite{low2019hamiltonian,gilyen2019quantum}, the phase angles that correspond to the approximation of a polynomial of fixed degree $d$ are found by optimization methods that minimize the difference between the found polynomial and the target function. In particular, the optimizer we utilize finds symmetrical phase angles via an iterative quasi-Newton method~\citep{dong2021efficient}. Similar research using optimization methods to find phase angles for QSP and QSVT includes \citep{chao2020finding, wang2022energy}.

To apply the angles from QSP to those used in QSVT, one has to first convert the angles from the $W_x$ convention used in QSP to the reflection convention of block encoding construction used in QSVT. We modify the phase angles as follows, where $\phi_k'$ are the angles obtained for the $W_x$ convention, and $\phi_k$ are the angles for the reflection convention.
\begin{equation}
    \phi_k=\begin{cases}
        \phi_0' + \frac{\pi}{4} - \frac{d \pi}{2} & k = 0          \\
        \phi_k' + \frac{\pi}{2}                   & 1\le k \le d-1 \\
        \phi_d' + \frac{\pi}{4}                   & k=d            \\
    \end{cases}.
\end{equation}

Given that the obtained angles from QSP/QSVT implements $g \in \mathbb{C}[x]$ such that $f(x) = \mathrm{Re}(g(x))$, to implement the block encoding of $f(D)$, we need to implement $f(x) = \frac{1}{2}(g(x)+g^*(x))$ where the phase angles of $g^*(x)$ are obtained by negating the phase angles obtained for $g(x)$. While this can be achieved by simply taking the sum of two QSVT circuits by LCU, one can alternatively modify the QSVT circuit such that a Hadamard gate is added to the ancilla qubit for the phase-controlled rotation gates before and after the QSVT circuit, as shown in \cref{figRealQET}, as the additional input of the $\ket{1}$ state would negate the phase angles, as shown in \citep{martyn2021grand}.

\begin{figure}
    \centering
    \includegraphics[width=0.5\linewidth]{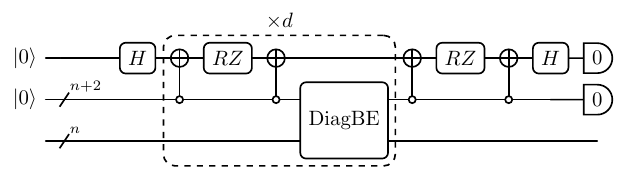}
    \caption{Circuit for implementation of a matrix polynomial.}
    \label{figRealQET}
\end{figure}

For classical simulations, we first generate the original probability distribution, then sample from the distribution. From the samples, we compute the sample distribution and truncate samples whose occurrences in the sample distribution fall above the threshold -- thus finding the set of rare events that can be later used for importance/rejection sampling. Note that when one increases classical queries, most queries would produce non-rare events, which would increase the TVD until the increased queries would allow for an extra sample of rare events, and the TVD drops again. This results in a near-periodic fluctuation of the TVD in the classical sampling method that is absent in the quantum method.

Our simulation results show that our algorithm can be used to effectively identify rare events, suppress non-rare events, and linearly boost the probabilities of the rare events, given that we can approximate the threshold function to a sufficient closeness.

\begin{figure}
    \centering
    \begin{tikzpicture}[>=stealth, node distance=3.4cm]
        \node[draw, circle, minimum size=1cm] (h) {$s_0$};
        \node[draw, circle, minimum size=1cm, right of=h] (t) {$s_1$};

        \draw[->, bend left=20] (h) to node[above] {$p,\ 1$} (t);
        \draw[->, bend left=20] (t) to node[below] {$p,\ 0$} (h);
        \draw[->, out=135, in=45, looseness=7] (h) to node[above] {$1-p,\ 0$} (h);
        \draw[->, out=135, in=45, looseness=7] (t) to node[above] {$1-p,\ 1$} (t);
    \end{tikzpicture}
    \caption{Perturbed coin dynamics. State $s_0$ denotes heads and state $s_1$ denotes tails. At each step, the coin flips with probability $p$ and otherwise remains in the current state.}
    \label{fig:pcoin}
\end{figure}
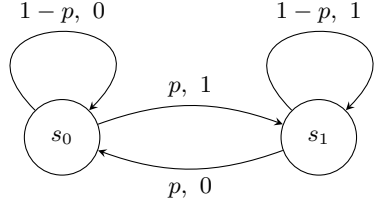

\subsection{Perturbed coin model}
\label{app:q_pcoin}

\begin{figure}
    \includegraphics[width=\linewidth]{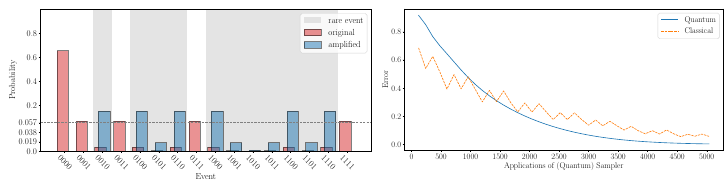}
    \caption{\textbf{Simulation results for a perturbed coin.} The left panel shows the output of the quantum algorithm on a smaller-sized version of the p-coin simulation with $L=4$. The new distribution of rare events (blue) amplifies rare events while non-rare events are suppressed compared to the original (red). The right panel shows the convergence of the TVD error for $L=8$ as the number of queries increases. The quantum algorithm converges faster and more smoothly compared to its classical counterpart.}
    \label{fig:results}
\end{figure}

We apply our example to a simple Markovian process that can be generated by a perturbed coin. At each time step, the coin flips with probability $p$, otherwise remains at its original state, as shown in~\cref{fig:pcoin}.
As time evolves, one may observe a sequence of heads (denoted as 0) and tails (denoted as 1). The classical perturbed coin contains only two classical memory states: heads (denoted as $s_0$) and tails (denoted as $s_1$), which carry information from the past.

Quantum models use the quantum memory states $\{\ket{\sigma_0},\ket{\sigma_1}\}$ to carry the information from past rather than the classical states~\cite{felix2018Practical,yang2018matrix}. These quantum states can be expressed as
\begin{equation}
    \begin{aligned}
        \ket{\sigma_0} & = \sqrt{1-p}\ket{0} + \sqrt{p}\ket{1}, \\
        \ket{\sigma_1} & = \sqrt{p}\ket{0} + \sqrt{1-p}\ket{1}, \\
    \end{aligned}
\end{equation}
$\ket{\sigma_0}$ and $\ket{\sigma_1}$ are non-orthogonal to each other, which leads to the memory advantage over the classical model.

The unitary operator $U$ implements the transition between the quantum memory states. At each time step, the quantum model implements the coupling unitary operator $U$ to couple the memory state $\ket{\sigma_i}$ and the output register, initialized in the vacuum state $\ket{0}$. Since the perturbed coin is a Markov process, the output $0/1$ determines the next memory state to be $\ket{\sigma_{0/1}}$. The unitary operator $U$ satisfies
\begin{equation}
    \begin{aligned}
        A_0 & = \bra{0}U\ket{0} = \ket{\sigma_0}\bra{0}, \\
        A_1 & = \bra{1}U\ket{0} = \ket{\sigma_1}\bra{1}.
    \end{aligned}
\end{equation}
The unitary operator $U$ can be constructed from the Kraus operators $A_x$. Specifically, we choose the first qubit to be the memory qubit and the second qubit to be the output qubit. The matrix form of $U$ is
\begin{equation}
    U = \begin{bmatrix}
        \sqrt{1-p} & * & 0          & * \\
        0          & * & \sqrt{p}   & * \\
        \sqrt{p}   & * & 0          & * \\
        0          & * & \sqrt{1-p} & *
    \end{bmatrix}
\end{equation}
where $*$ denotes the missing entries. The unitary operator can then be obtained by using the Gram-Schmidt procedure to obtain the missing entries.

We encode the probability distribution in the amplitudes of a quantum state and block-encode the distribution as a diagonal block-encoding. We set sequence length $L=8$, rareness parameter $K=2.2$, and the probability of the coin to be $p=0.1$. The history is set to be $S_0$. We evaluate the TVD distance of the sampling distribution. To compare with classical Monte Carlo methods of simulation, we fix the number of degrees $d$ for polynomial approximation for the quantum simulations, which would amount to $ d\lceil 1/\sqrt{p_{\rm rare}}\rceil$ queries after $\lceil1/\sqrt{p_{\rm rare}}\rceil$ rounds of amplitude amplification. We disregard the final post-selection of the quantum case and directly take the probability distribution that would be generated from the quantum state if post-selection were 0 after amplification. In the classical case, we match the number of queries $ d\lceil 1/\sqrt{p_{\rm rare}}\rceil$ to identify rare events, and further, to compete on fair grounds, we do not perform rejection sampling in the importance sampling phase, but only account for the runtime to identify the rare events. Further implementation details can be found in \cref{app:numerics}. We plot the TVD distance of our quantum algorithm together with that of the Monte Carlo method in~\cref{fig:results}.
\end{document}